\documentclass[11pt]{article}

\usepackage[letterpaper,margin=1in]{geometry}
\usepackage[T1]{fontenc}
\usepackage{lmodern}
\usepackage{amsmath,amssymb,amsthm,mathtools}
\usepackage{enumitem,microtype}
\usepackage{aliascnt}
\usepackage{xcolor}

\usepackage[colorlinks,allcolors=blue,bookmarksdepth=2]{hyperref}
\hypersetup{
  pdftitle={SVP Is NP-Hard for Some Rank-2 Cyclotomic Modules},
  pdfauthor={Jiaqi Liu, Yansong Feng, Yanbin Pan},
  pdfsubject={Exact shortest vector problem on rank-two cyclotomic modules},
  pdfkeywords={shortest vector problem, module lattices, cyclotomic fields, NP-completeness}
}

\allowdisplaybreaks

\newcommand{\Z}{\mathbb Z}
\newcommand{\Q}{\mathbb Q}
\newcommand{\F}{\mathbb F}
\newcommand{\R}{\mathbb R}
\newcommand{\C}{\mathbb C}
\newcommand{\NP}{\mathsf{NP}}
\newcommand{\SVP}{\mathsf{SVP}}
\newcommand{\dist}{\operatorname{dist}}

\newtheorem{theorem}{Theorem}[section]
\newaliascnt{lemma}{theorem}
\newtheorem{lemma}[lemma]{Lemma}
\aliascntresetthe{lemma}
\newaliascnt{corollary}{theorem}
\newtheorem{corollary}[corollary]{Corollary}
\aliascntresetthe{corollary}
\newaliascnt{proposition}{theorem}
\newtheorem{proposition}[proposition]{Proposition}
\aliascntresetthe{proposition}
\theoremstyle{definition}
\newaliascnt{definition}{theorem}
\newtheorem{definition}[definition]{Definition}
\aliascntresetthe{definition}
\theoremstyle{plain}

\usepackage[nameinlink,capitalize]{cleveref}
\crefname{lemma}{Lemma}{Lemmas}
\Crefname{lemma}{Lemma}{Lemmas}
\crefname{corollary}{Corollary}{Corollaries}
\Crefname{corollary}{Corollary}{Corollaries}
\crefname{proposition}{Proposition}{Propositions}
\Crefname{proposition}{Proposition}{Propositions}
\crefname{definition}{Definition}{Definitions}
\Crefname{definition}{Definition}{Definitions}
\crefname{equation}{Equation}{Equations}
\Crefname{equation}{Equation}{Equations}

\begin{document}

\hypersetup{pageanchor=false}
\begin{titlepage}
\thispagestyle{empty}

\begin{center}
  \vspace*{1.5cm}
  {\LARGE\bfseries
SVP Is NP-Hard for Some Rank-2 Cyclotomic Modules
    \par}
\end{center}

\vspace{0.55cm}

\begin{center}
  {\large Jiaqi Liu \quad Yansong Feng \quad Yanbin Pan\par}
  \vspace{0.8em}
  {\small
    State Key Laboratory of Mathematical Sciences,\par
    Academy of Mathematics and Systems Science, Beijing, China\par}
  \vspace{0.35em}
  {\small\texttt{\{ljqi,fengyansong,panyanbin\}@amss.ac.cn}\par}
\end{center}

\vspace{0.55cm}

\begin{abstract}
Let $q$ be an odd prime. Let $\zeta_q$ be a primitive $q$th root of unity, and put $K=\Q(\zeta_q)$, with ring of integers $\mathcal{O}_K=\Z[\zeta_q]$.  As $q$ ranges over odd primes, we prove that the decision version of the Shortest Vector Problem ($\SVP$) in the canonical $\ell_2$-norm is $\NP$-complete on full-rank free submodules of $\mathcal{O}_K^2$ by giving a deterministic polynomial-time many-one reduction from Exact Cover by 3-Sets (X3C).  The module rank is fixed at two.  The main obstacle is closure under the action of $\mathcal{O}_K$.  A module containing a nonzero vector also contains every scalar multiple of that vector by a nonzero element of $\mathcal{O}_K$, and some of these multiples may be shorter.

Three ideas overcome this obstacle.  First, we map the Bennett--Peikert Reed--Solomon lattice to a principal cyclotomic ideal and use Wan's point-count estimates to prove that a coset of this ideal contains many binary coefficient representatives. Second, we construct a checker whose canonical squared norm encodes the X3C equations, with a controlled nonnegative error term.  Third, the checker and a second module coordinate combine with a separation bound for ideal cosets to rule out short vectors in NO instances. Each constructed instance consists of an odd prime $q$, two integral generators whose $2\times2$ generator matrix has nonzero determinant, and an integer squared threshold.  The construction also gives $\NP$-hardness of search-$\SVP$ under deterministic polynomial-time Turing reductions with at most one oracle call.
\end{abstract}
\end{titlepage}

\hypersetup{pageanchor=true}
\pagenumbering{arabic}
\setcounter{page}{1}

\section{Introduction}
\begingroup
\crefname{lemma}{Lemma}{Lemmas}
\Crefname{lemma}{Lemma}{Lemmas}
\crefname{corollary}{Corollary}{Corollaries}
\Crefname{corollary}{Corollary}{Corollaries}
\crefname{proposition}{Proposition}{Propositions}
\Crefname{proposition}{Proposition}{Propositions}
A lattice $\mathcal L\subseteq\R^m$ is the set of all integer linear combinations of linearly independent vectors.  Given a basis of $\mathcal L$, the decision shortest vector problem ($\SVP$) asks whether the lattice contains a nonzero vector of squared norm at most a given threshold.  Its search version asks for a shortest nonzero lattice vector. In 1981, Van Emde Boas asked whether exact $\SVP$ is $\NP$-hard~\cite{vanEmdeBoas81}.  Later, Ajtai's breakthrough result established that exact Euclidean $\SVP$ is $\NP$-hard under randomized reductions~\cite{Ajtai98}.  For Euclidean $\SVP$ on lattices of rank $n$, later randomized reductions established hardness of approximation within a factor of $1+n^{-\varepsilon}$ for every fixed $\varepsilon>0$~\cite{cai1998approximating}, within any fixed factor in $(1,\sqrt2)$~\cite{Micciancio01}, and within every constant factor~\cite{Khot05,Micciancio12}. Stronger randomized-time assumptions yield superconstant and almost-polynomial factors~\cite{Khot05,HavivRegev12}.

However, a deterministic $\NP$-hardness proof for exact Euclidean $\SVP$ remained open for decades. Along one route, Hair and Sahai used probabilistically checkable proofs (PCPs) to obtain deterministic hardness for $\ell_p$-norms with $p>2$~\cite{HairSahai26STOC}. They also proved approximation hardness in the $\ell_p$-norm for every fixed $1\le p<\infty$ under a subexponential-time assumption~\cite{HairSahai26FOCS}. In two recent preprints, Wan proved deterministic hardness for every constant approximation factor in the Euclidean norm~\cite{wan2026euclidean}, extending an earlier result for factors below $\sqrt2$, including the exact problem~\cite{wan2026np}.

At an intuitive level, Micciancio's \emph{local-density} program asks for an explicit lattice $\mathcal L$ and a coset $\mathbf t+\mathcal L$ containing many points of norm strictly below $\lambda_1(\mathcal L)$~\cite{Micciancio01,Micciancio12}.  The minimum of $\mathcal L$ protects the soundness gap, while the abundance of short coset points allows the reduction to encode a solution to the input problem.  For our deterministic reduction, these points must realize every prescribed assignment on a designated set of coordinates.  Bennett and Peikert proposed a candidate by lifting Reed--Solomon codes~\cite{BennettPeikert23}.  Wan used finite-field point counts to prove that an explicit Reed--Solomon coset has the required representatives with prescribed coordinates~\cite{wan2026np}. Wan's result supplies the local-density statement needed for deterministic hardness results for general lattices in the $\ell_2$-norm.  It does not by itself show that the reduction remains sound when the output is required to have a prescribed cyclotomic module structure.

Under the canonical embedding of a number field, ideals yield ideal lattices~\cite{lyubashevsky2013ideal}.  Higher-rank modules similarly yield module lattices~\cite{LangloisStehle15,FelderhoffPelletMaryStehle22}. The ring action supports compact representations and fast arithmetic~\cite{HoffsteinPipherSilverman98,LyubashevskyPeikertRegev13}, underlies Ring-LWE and Module-LWE~\cite{lyubashevsky2013ideal,LangloisStehle15}, and enables specialized algorithms for module lattices~\cite{LeePelletMaryStehleWallet19}. This algebraic structure is no longer just a theoretical design choice.  In 2024, NIST standardized ML-KEM (FIPS~203), a module-lattice-based scheme derived from CRYSTALS--Kyber and based on Module-LWE~\cite{BosDucasKiltzEtAl18,NISTFIPS203}.  It also standardized ML-DSA (FIPS~204), derived from CRYSTALS--Dilithium and relying on Module-LWE and a variant of Module-SIS~\cite{DucasKiltzLepointEtAl18,NISTFIPS204}.

The cryptographic importance of module lattices has motivated broad work on structured lattices.  Prior work studies cyclic-lattice problems~\cite{Micciancio07}, module unique-SVP and NTRU~\cite{FelderhoffPelletMaryStehle22}, probabilistic reductions for approximation problems from arbitrary modules to free modules~\cite{DeMicheliMicciancioPelletMaryTran23}, worst-case to average-case reductions for the approximate shortest independent vectors problem ($\mathsf{SIVP}$) at fixed module rank under the Extended Riemann Hypothesis~\cite{deBoerPageTomaWesolowski26}, and the geometry of shortest vectors in random modules~\cite{GargavaSerbanViazovskaViglino25}.

For the Closest Vector Problem ($\mathsf{CVP}$), a recent preprint proves that exact Euclidean decision-$\mathsf{CVP}$ is $\NP$-complete on coefficient lattices of nonzero principal ideals in power-of-two cyclotomic rings, with a lift to principal cyclic ideal lattices~\cite{LiuFengPan26}. To the best of our knowledge, no deterministic polynomial-time $\NP$-hardness result for exact Euclidean $\SVP$ on full-rank free rank-two modules over rings of integers in cyclotomic fields was known before the first version of this work appeared.

\enlargethispage{\baselineskip}
\begin{samepage}
This leads to the question addressed in this work:
\begin{quote}
\emph{Is Euclidean $\SVP$ $\NP$-hard for rank-two modules over cyclotomic rings?}
\end{quote}
We answer this question affirmatively, even for full-rank free submodules of $\mathcal{O}_K^2$ with $K=\Q(\zeta_q)$ for odd primes $q$.
\end{samepage}

\subsection{Our result}

For every odd prime $q$, let $\zeta:=\zeta_q$ be a primitive $q$th root of unity.  Write $K:=\Q(\zeta)$ and $\mathcal{O}_K:=\Z[\zeta]$.  For $\mathbf y=(y_1,y_2)\in K^2$, define the canonical product norm by
\[
 \|\mathbf y\|^2
 :=\sum_{\sigma:K\hookrightarrow\C}
   \bigl(|\sigma(y_1)|^2+|\sigma(y_2)|^2\bigr).
\]
We write $\|\cdot\|_2$ for the $\ell_2$-norm of coefficient vectors.  An $\mathcal{O}_K$-submodule $\mathcal M\subseteq\mathcal{O}_K^2$ is called \emph{full rank} if it has finite index in $\mathcal{O}_K^2$.

We consider the following decision problem.  An input consists of an odd prime $q$, two vectors $\mathbf m_1,\mathbf m_2\in\mathcal{O}_K^2$ with $\det[\,\mathbf m_1\ \mathbf m_2\,]\ne0$, and an integer $S\ge0$.  Each of the four ring coordinates of $\mathbf m_1$ and $\mathbf m_2$ is specified by its $q-1$ integer coefficients in the basis $1,\zeta,\ldots,\zeta^{q-2}$.  For $\mathcal M:=\mathcal{O}_K\mathbf m_1+\mathcal{O}_K\mathbf m_2$, the question is whether
\[
 \lambda_1(\mathcal M)^2\le S.
\]
We call this problem \emph{decision-$\SVP$ in the $\ell_2$-norm on rank-two cyclotomic modules}.  The corresponding search-$\SVP$ problem takes $q,\mathbf m_1,\mathbf m_2$ as input and asks for a shortest nonzero vector of $\mathcal M$.  The determinant condition makes $\mathcal M$ a full-rank free submodule of $\mathcal{O}_K^2$ generated by two vectors.

\begin{theorem}
\label{thm:main-intro}
As $q$ ranges over odd primes, decision-$\SVP$ in the $\ell_2$-norm on free rank-two cyclotomic modules is $\NP$-complete.  The hardness holds under deterministic polynomial-time many-one reductions.
\end{theorem}

Our reduction is from Exact Cover by 3-Sets (X3C), which is $\NP$-complete~\cite{GareyJohnson79}.  Given an X3C instance, the reduction outputs a tuple $(q,\mathbf m_1,\mathbf m_2,S)$ of the form above. Here $q$ is an odd prime, $\mathbf m_1,\mathbf m_2\in\mathcal{O}_K^2$ generate the output module and have nonzero determinant, and $S$ is an integer squared threshold. As a lattice, the output module has $\Z$-rank $2(q-1)$. The tuple has polynomial encoding length. \Cref{subsec:encoding} gives the choice of $q$ and bounds the encoding length and running time.

This is a worst-case hardness result in which the cyclotomic ring varies with $q$.  The theorem does not establish hardness for the power-of-two cyclotomic family, on average, or through a cryptographic security reduction. The same construction also gives $\NP$-hardness of search-$\SVP$.

\begin{corollary}
\label{cor:search-hardness}
As $q$ ranges over odd primes, search-$\SVP$ in the $\ell_2$-norm on free rank-two submodules of $\mathcal{O}_K^2$ is $\NP$-hard.  The hardness holds under deterministic polynomial-time Turing reductions with at most one oracle call.
\end{corollary}

\paragraph{Subsequent work.}
Wan proved deterministic $\NP$-completeness for exact decision-$\SVP$ on full-rank cyclic lattices in the coefficient $\ell_2$-norm~\cite{wan2026cyclic}. Martin gave a dimension-preserving deterministic polynomial-time reduction from general Euclidean $\SVP$ to ideal lattices over monogenic orders in totally real number fields under the canonical embedding~\cite{Martin26}.

\subsection{Proof overview}
Our goal is to give a deterministic polynomial-time reduction from X3C to decision-$\SVP$ on rank-two modules over cyclotomic rings.  The reduction directly constructs a rank-two $\mathcal{O}_K$-module $\mathcal M\subseteq\mathcal{O}_K^2$ and an integer squared threshold $S$ such that
\[
 \lambda_1(\mathcal M)^2\le S
 \quad\Longleftrightarrow\quad
 \text{the X3C instance is positive}.
\]
The main challenge is to prove the strict lower bound $\lambda_1(\mathcal M)^2>S$ for every negative instance despite closure under multiplication by $\mathcal{O}_K$.

\begin{samepage}
X3C asks whether a given matrix $\mathbf A\in\{0,1\}^{M\times n}$ with exactly three ones in each column admits a vector $\boldsymbol\xi\in\{0,1\}^n$ satisfying
\begin{equation}
\label{eq:intro-x3c}
 \mathbf A\boldsymbol\xi=\mathbf1_M,
\end{equation}
where $\mathbf1_M$ is the all-one vector in $\{0,1\}^M$.
\end{samepage}

The proof of the reduction has three layers.  First, we show that every $\boldsymbol\xi$ has an extension to a binary coefficient vector $\mathbf x$ whose evaluation at $\zeta$ lies in a fixed ideal coset in $\mathcal{O}_K$.  Second, we construct an X3C checker whose squared norm encodes $\mathbf A\boldsymbol\xi-\mathbf1_M$.  Third, we place the checker in a rank-two module.  The reduction computes its output directly from $\mathbf A$ without finding $\mathbf x$.

\paragraph{The ideal coset.}
The reduction embeds $\boldsymbol\xi\in\{0,1\}^n$ in a longer binary coefficient vector $\mathbf x=(x_a)_{a\in\F_q}\in\{0,1\}^{\F_q}$ by choosing distinct field elements $\alpha_1,\ldots,\alpha_n\in\F_q$ and requiring $x_{\alpha_j}=\xi_j$ for $1\le j\le n$.  Put $\pi:=1-\zeta$. For a chosen integer $k$, define $\mathfrak a:=\pi^k\mathcal{O}_K$.  For $\mathbf x=(x_a)_{a\in\F_q}\in\Z^{\F_q}$, evaluation at $\zeta$ gives the corresponding ring element
\[
 [\mathbf x]:=\sum_{a\in\F_q}x_a\zeta^a\in\mathcal{O}_K.
\]

By \cref{thm:cyclotomic-descent}, which uses the Bennett--Peikert minimum bound~\cite{BennettPeikert23}, evaluation maps a Reed--Solomon lattice to the ideal lattice $\mathfrak a$ and $\lambda_1(\mathfrak a)^2\ge2kq$.  Wan's point-count estimate~\cite{wan2026np}, in the uniform form of \cref{thm:uniform-completions}, supplies many binary vectors whose corresponding ring elements lie in a fixed coset of $\mathfrak a$.  The reduction chooses one coset $c+\mathfrak a$ with $c\in\mathcal{O}_K$.  By \cref{lem:coset-syndrome,thm:uniform-completions}, every $\boldsymbol\xi\in\{0,1\}^n$ has many extensions $\mathbf x\in\{0,1\}^{\F_q}$ such that
\[
 [\mathbf x]\in c+\mathfrak a,
 \qquad \sum_{a\in\F_q}x_a=h,
 \qquad x_{\alpha_j}=\xi_j\quad(1\le j\le n),
\]
where $h$ is the prescribed Hamming weight. Thus $x_{\alpha_j}=\xi_j$ stores $\boldsymbol\xi$ in designated coordinates $\alpha_1,\ldots,\alpha_n$ of $\mathbf x$, while $[\mathbf x]$ places the resulting ring element in the coset $c+\mathfrak a$. For the chosen parameters, $[\mathbf x]$ has squared canonical norm $H=h(q-h)<2kq$ and is therefore shorter than every nonzero element of $\mathfrak a$.

\paragraph{The X3C checker.}
The second step constructs $U,V\in\mathcal O_K$ so that the squared canonical norm of $U[\mathbf x]-V$ encodes the residual $\mathbf A\boldsymbol\xi-\mathbf1_M$ of \cref{eq:intro-x3c}. By \cref{thm:checker-score}, every weight-$h$ binary vector $\mathbf x$ with $\xi_j=x_{\alpha_j}$ satisfies
\begin{equation}
 \|U[\mathbf x]-V\|^2
 =Q^2\bigl(B_0+2\kappa\|\mathbf A\boldsymbol\xi-\mathbf1_M\|_2^2
   +\delta(\mathbf x)\bigr).
\end{equation}
Here $B_0$ and the positive parameters $Q,\kappa$ are independent of $\mathbf x$, and $\delta(\mathbf x)\ge0$ is the total error.  We call a representative $\mathbf x$ \emph{admissible} if it satisfies the correlation conditions in \cref{sec:clean-representatives}, which ensure $\delta(\mathbf x)<\kappa/8$. The uniform count in \cref{thm:uniform-completions} and the union bound in the proof of \cref{thm:enhanced-clean} show that every $\boldsymbol\xi\in\{0,1\}^n$ has an admissible representative $\mathbf x$ with $[\mathbf x]\in c+\mathfrak a$. Consequently, an exact cover gives an admissible representative whose squared checker norm is at most $Q^2(B_0+\kappa/8)$.  For a negative instance, every weight-$h$ binary representative has squared checker norm at least $Q^2(B_0+2\kappa)$.

\paragraph{The rank-two module.}
The third step places the value $U[\mathbf x]-V$ in a module.  Choose a positive integer $\Gamma$ to control the second module coordinate.  Define the rank-two $\mathcal{O}_K$-submodule of $\mathcal O_K^2$
\[
 \mathcal M
 :=\mathcal{O}_K(U\pi^k,0)
   +\mathcal{O}_K(Uc-V,-Q\Gamma)
 \subset\mathcal O_K^2.
\]
Every vector in $\mathcal M$ has the unique form
\begin{equation}
\label{eq:intro-module-vector}
 \mathbf v(y,w)
 :=\bigl(Uy-Vw,-Q\Gamma w\bigr),
 \qquad y,w\in\mathcal{O}_K,\quad y-wc\in\mathfrak a.
\end{equation}
When $w=1$ and $y=[\mathbf x]$, the first coordinate is the checker value $U[\mathbf x]-V$.  The second coordinate contributes $Q^2\Gamma^2\|w\|^2$ to the squared norm. The integer threshold in \cref{eq:F-S-nu-exact} satisfies $S=Q^2(B_0+F+\kappa)$, where $F=\Gamma^2(q-1)$. If $\boldsymbol\xi$ is an exact cover, choose an admissible representative $\mathbf x$.  Taking $w=1$ in \cref{eq:intro-module-vector} then gives $\|\mathbf v([\mathbf x],1)\|^2<S$.

\paragraph{Soundness under multiplication.}
Suppose that the X3C instance is negative, and consider a nonzero vector $\mathbf v(y,w)\in\mathcal M$.  The main issue is possible cancellation between $Uy$ and $Vw$.  Since $y\in wc+\mathfrak a$, one has $\|y\|\ge\dist(wc,\mathfrak a)$. For an explicit $\nu=o(1)$, \cref{lem:gap-and-stability} gives
\begin{equation}
\label{eq:intro-stability}
 \|\mathbf v(y,w)\|^2
 \ge Q^2(1-\nu)\bigl(\|y\|^2+\Gamma^2\|w\|^2\bigr).
\end{equation}
This inequality lets us use lower bounds for $\|y\|$ and $\|w\|$ separately and transfers the geometric bounds from the ideal cosets to the module.  The argument considers four cases.  The comparisons with $S$ below are verified in the proof of \cref{thm:exact-X3C-encoding}.
\begin{enumerate}[label=(\roman*),leftmargin=*]
\item If $w=0$, then $0\ne y\in\mathfrak a$. \Cref{lem:checker-U-properties} gives $\|Uy\|\ge Q\|y\|$, so the minimum bound for $\mathfrak a$ puts the squared norm of every nonzero module vector with $w=0$ above $S$.
\item If $w$ is a root of unity, write $y=wu$ with $u\in c+\mathfrak a$.  Then $\mathbf v(y,w)=w\mathbf v(u,1)$, and multiplication by $w$ preserves the canonical norm. Let $\mathbf x\in\Z^{\F_q}$ be the unique integral coefficient representative of $u$ with coefficient sum $h$ from \cref{lem:coset-syndrome}. If $\mathbf x$ is binary, the checker bound for a negative instance and the second coordinate put the squared norm of the module vector above $S$. If $\mathbf x$ is not binary, \cref{lem:coset-syndrome} gives $\|u\|^2\ge H+2q$.  Applying \cref{eq:intro-stability} to $u$ and $1$ again puts the squared norm of the module vector above $S$.
\item If $w$ is not a root of unity and $0<\|w\|^2\le3q$, \cref{thm:robust-center} gives a lower bound for $\dist(wc,\mathfrak a)$, and part~(ii) of \cref{lem:sparse-tools} gives $\|w\|^2\ge2q-4$. Combining these bounds with \cref{eq:intro-stability} puts the squared norm of the module vector above $S$.
\item If $w$ is not a root of unity and $\|w\|^2>3q$, the second coordinate $-Q\Gamma w$ alone puts the squared norm of the module vector above $S$.
\end{enumerate}

\paragraph{Organization.}
\Cref{sec:preliminaries} recalls Euclidean lattices, cyclotomic rings, X3C, and Sidon sequences, and records the power-sum estimates used to count binary representatives in ideal cosets. \Cref{sec:descent} maps the lifted Reed--Solomon lattice to a principal cyclotomic ideal, proves the fixed-coordinate count, and establishes the norm-separation bounds used for soundness. \Cref{sec:x3c-checker} constructs the X3C checker and proves the existence of admissible representatives.  \Cref{sec:reduction} gives the main reduction.  It assembles the rank-two module, proves completeness and soundness, and bounds the encoding length and running time.

\endgroup

\section{Preliminaries}
\label{sec:preliminaries}
\paragraph{Notation.}
For a positive integer $r$, write $[r]:=\{1,\ldots,r\}$.  For a prime $q$, let $\F_q$ denote the finite field with $q$ elements, and let $\F_q^\times:=\F_q\setminus\{0\}$ denote its multiplicative group. When an element of $\F_q$ occurs as an exponent or in a binomial coefficient, we use its unique integer representative in $\{0,\ldots,q-1\}$. Vectors and matrices are represented in bold, e.g., $\mathbf v$ for a vector and $\mathbf A$ for a matrix, while their scalar coordinates are not. We use $\|\cdot\|_2$ for the $\ell_2$-norm and $\mathbf1_r$ for the vector of ones in $\R^r$. The notation $\mathfrak a\triangleleft A$ means that $\mathfrak a$ is an ideal of the ring $A$.

\subsection{Euclidean lattices}
A Euclidean lattice of rank $n$ is a discrete additive subgroup $\mathcal L\subseteq\R^m$ of the form
\[
 \mathcal L=\mathcal L(\mathbf B):=\mathbf B\Z^n,
\]
where $\mathbf B\in\R^{m\times n}$ has linearly independent columns.  These columns form a basis. The first minimum of $\mathcal L$ is the length of a shortest nonzero lattice vector:
\[
 \lambda_1(\mathcal L)
 =\min_{\mathbf0\ne\mathbf v\in\mathcal L}\|\mathbf v\|_2.
\]

\begin{definition}[SVP]
The \emph{decision-$\SVP$} problem takes as input a basis $\mathbf B\in\Q^{m\times n}$ with linearly independent columns and a rational squared threshold $\Delta\geq 0$, and asks whether
\[
 \lambda_1(\mathcal L(\mathbf B))^2\le\Delta.
\]
\end{definition}
The problem belongs to $\NP$: a YES certificate is a nonzero integer coefficient vector for a short lattice point.  Standard cofactor bounds give a certificate of polynomial bit length, and exact rational arithmetic checks both membership and the squared norm. See, e.g.,~\cite{MicciancioGoldwasser02}. The corresponding search-$\SVP$ problem takes a full-column-rank basis $\mathbf B$ and outputs a shortest nonzero vector of $\mathcal L(\mathbf B)$.

\subsection{Cyclotomic fields and module lattices}

\paragraph{Cyclotomic field.}
In this paper, we only consider odd primes $q$.  For a fixed $q$, let $\zeta:=\zeta_q\in\C$ be a primitive $q$th root of unity.  We suppress the subscript $q$ after fixing the prime and put
\[
 K:=\Q(\zeta),
 \qquad
 \mathcal{O}_K:=\Z[\zeta],
 \qquad
 \pi:=1-\zeta.
\]
Here $\mathcal{O}_K$ is the ring of integers of $K$.  The elements $1,\zeta,\ldots,\zeta^{q-2}$ form an integral basis of $K$, so $[K:\Q]=q-1$.  The prime $q$ is totally ramified in $\mathcal{O}_K$, and hence $(q)=(\pi)^{q-1}$~\cite{Washington97}.

\paragraph{Canonical embedding.}
For $v\in K$, write $\overline v$ for its complex conjugate.  In particular, $\overline\zeta=\zeta^{-1}$.  For $a\in\F_q^\times$, let $\sigma_a:K\hookrightarrow\C$ be the embedding defined by $\sigma_a(\zeta)=\zeta^a$.  These are all the embeddings of $K$. The canonical embedding sends $v$ to $(\sigma_a(v))_{a\in\F_q^\times}$.  For each integer $d\ge1$ and vectors $\mathbf y=(y_1,\ldots,y_d),\mathbf y'=(y'_1,\ldots,y'_d)\in K^d$, define the real inner product
\[
 \langle\mathbf y,\mathbf y'\rangle
 :=\operatorname{Re}\sum_{i=1}^{d}\sum_{a\in\F_q^\times}
   \sigma_a(y_i)\overline{\sigma_a(y'_i)}
\]
and its induced product norm
\[
 \|\mathbf y\|^2:=\langle\mathbf y,\mathbf y\rangle
 =\sum_{i=1}^{d}\sum_{a\in\F_q^\times}|\sigma_a(y_i)|^2.
\]
For $d=1$, this is the canonical norm on $K$:
\[
 \|v\|^2=\sum_{a\in\F_q^\times}|\sigma_a(v)|^2.
\]
The roots of unity in $K$ are exactly $\pm\zeta^j$ for $0\le j<q$~\cite{Washington97}.  If $\varepsilon$ is a root of unity, then $\|\varepsilon\mathbf y-\varepsilon\mathbf y'\|=\|\mathbf y-\mathbf y'\|$ for all $\mathbf y,\mathbf y'\in K^{d}$.  Thus multiplication by $\varepsilon$ is an isometry.  In particular, it preserves distances.

\paragraph{Module lattices.}
For an integer $s\ge1$, an \emph{$\mathcal{O}_K$-module lattice} in $K^s$ is a finitely generated $\mathcal{O}_K$-submodule $\mathcal M\subseteq K^s$, equipped with the componentwise canonical embedding and the product norm.  It has \emph{full rank} if its $K$-span is $K^s$.  When $s=1$, a nonzero ideal $\mathfrak a\subseteq\mathcal{O}_K$ gives a module lattice in $K$ called an \emph{ideal lattice}.  Both module and ideal lattices are discrete additive subgroups of their real spans.

In this paper, we only use free module lattices of full rank.  If $\mathbf m_1,\ldots,\mathbf m_s\in K^s$ are linearly independent over $K$, then
\[
 \mathcal M=\mathcal{O}_K\mathbf m_1+\cdots+\mathcal{O}_K\mathbf m_s
\]
is free of module rank $s$, has full rank, and has $\Z$-rank $s(q-1)$.  The square matrix with columns $\mathbf m_1,\ldots,\mathbf m_s$ has nonzero determinant precisely when these vectors are linearly independent over $K$.  If $\mathcal M\subseteq\mathcal{O}_K^s$, then full rank is equivalent to finite index in $\mathcal{O}_K^s$. Under the canonical embedding, the product norm is the Euclidean $\ell_2$-norm.  We specify the lattice by its module generators in power-basis coordinates.

The first minimum of an ideal lattice or module lattice is the length of its shortest nonzero vector in the canonical or product norm.  For $v\in K$ and an ideal lattice $\mathfrak a$, define
\[
 \dist(v,\mathfrak a):=\min_{y\in\mathfrak a}\|v-y\|.
\]

\paragraph{Coefficient representatives.}
For $\mathbf z=(z_a)_{a\in\F_q}\in\Q^{\F_q}$, write
\[
 [\mathbf z]:=\sum_{a\in\F_q}z_a\zeta^a\in K.
\]
If $\mathbf z\in\Z^{\F_q}$, then $[\mathbf z]\in\mathcal{O}_K$, and we call $\mathbf z$ an \emph{integral coefficient representative} of $[\mathbf z]$.  Since $\Phi_q(X)=1+X+\cdots+X^{q-1}$ is the minimal polynomial of $\zeta$, two integral coefficient vectors represent the same element precisely when they differ by an integer multiple of the all-one vector $\mathbf1_q$.

The following identity converts the canonical inner product and norm into calculations with these $q$ coefficients.
\begin{lemma}
\label{lem:cyclotomic-form}
For all rational coefficient vectors $\mathbf x,\mathbf y\in\Q^{\F_q}$,
\[
 \langle[\mathbf x],[\mathbf y]\rangle
 =q\sum_a x_ay_a-
 \left(\sum_a x_a\right)\left(\sum_a y_a\right).
\]
In particular, for every $\mathbf z\in\Z^{\F_q}$,
\[
 \|[\mathbf z]\|^2
 =q\sum_a z_a^2-\left(\sum_a z_a\right)^2
 =\sum_{a<b}(z_a-z_b)^2.
\]
Here $a<b$ uses the integer representatives $0,\ldots,q-1$.
\end{lemma}

\begin{proof}
For $a,b\in\F_q$, a finite geometric sum gives
\[
 \sum_{t\in\F_q^\times}\zeta^{t(a-b)}
 =\begin{cases}
   q-1,&a=b,\\
   -1,&a\ne b.
  \end{cases}
\]
Hence, for rational coefficient vectors $\mathbf x$ and $\mathbf y$,
\[
 \langle[\mathbf x],[\mathbf y]\rangle
 =q\sum_a x_ay_a-\sum_{a,b}x_ay_b
 =q\sum_a x_ay_a-\left(\sum_a x_a\right)\left(\sum_b y_b\right).
\]
Taking $\mathbf x=\mathbf y=\mathbf z$ gives the norm identity.  Expanding $\sum_{a<b}(z_a-z_b)^2$ gives its final form. The inner-product formula is unchanged under replacing $\mathbf x$ by $\mathbf x+r\mathbf1_q$ and $\mathbf y$ by $\mathbf y+s\mathbf1_q$, for $r,s\in\Q$.  Thus the identities do not depend on the choice of coefficient representatives.
\end{proof}

\subsection{X3C and Sidon sequences}
\label{sec:x3c-and-sidon}

An Exact Cover by 3-Sets (X3C) instance consists of a universe $[M]$ and three-element subsets $C_1,\ldots,C_n\subseteq[M]$.  Its incidence matrix $\mathbf A\in\{0,1\}^{M\times n}$ is defined by $A_{ij}=1$ when $i\in C_j$ and $A_{ij}=0$ otherwise.  A \emph{selection vector} $\boldsymbol\xi\in\{0,1\}^n$ specifies a choice of sets, with $\xi_j=1$ exactly when $C_j$ is selected.  We call $\boldsymbol\xi$ an \emph{X3C witness} if
\[
 \mathbf A\boldsymbol\xi=\mathbf1_M.
\]
This equality says that every universe element belongs to exactly one selected set.  The instance is positive exactly when it has an X3C witness, which specifies an exact cover of the universe. X3C is $\NP$-complete~\cite{GareyJohnson79}.

An instance with $M=0$ is positive. An instance whose incidence matrix has a zero row is negative. Henceforth, assume $M,n\ge1$ and that every row of $\mathbf A$ is nonzero. Then $M\le3n$.

To index the $n$ selection bits in the reduction, we choose distinct nonnegative integers whose nonzero ordered differences are all distinct.  Such a sequence is called a \emph{Sidon sequence}.  The following lemma gives a deterministic construction with largest entry $O(n^2)$.
\begin{lemma}
\label{lem:sidon}
For every $n\ge1$, one can deterministically compute integers
\[
 0\le\alpha_1<\cdots<\alpha_n=: \alpha_\star,
 \qquad \alpha_\star=O(n^2),
\]
such that all nonzero ordered differences $\alpha_\ell-\alpha_j$, $j\ne\ell$, are distinct. The computation takes time polynomial in $n$.
\end{lemma}

\begin{proof}
Use the deterministic quadratic-residue Sidon construction of Erd\H{o}s and Tur\'an~\cite{ErdosTuran41}.  For an explicit construction and the verification of the directed-difference property, see also Liu--Feng--Pan~\cite{LiuFengPan26}.
\end{proof}

\subsection{Power-sum estimates}

Let $q$ be prime, let $2\le k<m$, and let $\mathbf b=(b_1,\ldots,b_{k-1})\in\F_q^{k-1}$.  Define
\[
 \mathcal X_{k,m}(\mathbf b)
 :=\left\{\mathbf z=(z_1,\ldots,z_m)\in\F_q^m:
   \sum_{r=1}^m z_r^j=b_j\quad(1\le j<k)\right\}.
\]
The following proposition records the point-count bounds used below.

\begin{proposition}[{\cite{wan2026np}}]
\label{prop:wan-power-sum-estimates}
Let $q$ be prime and suppose that $2\le k\le m-2<q-2$.  For every $\mathbf b\in\F_q^{k-1}$, the following estimates hold.
\begin{enumerate}[label=(\roman*)]
 \item
 \[
  \left||\mathcal X_{k,m}(\mathbf b)|-q^{m-k+1}\right|
  \le \frac12(2k)^m q^{(m-k+2)/2}.
 \]
 \item For distinct $i,\ell\in[m]$,
 \[
  \left|
   \bigl|\{\mathbf z\in\mathcal X_{k,m}(\mathbf b):z_i=z_\ell\}\bigr|
   -q^{m-k}
  \right|
  \le \frac12(2k)^{m-1}q^{(m-k+2)/2}.
 \]
 \item For $i\in[m]$ and $a\in\F_q$,
 \[
  \left|
   \bigl|\{\mathbf z\in\mathcal X_{k,m}(\mathbf b):z_i=a\}\bigr|
   -q^{m-k}
  \right|
  \le \frac12(2k)^{m-1}q^{(m-k+1)/2}.
 \]
\end{enumerate}
\end{proposition}

\begin{proof}
Wan~\cite[Propositions~5.2 and~6.2]{wan2026np} proves part~(i) and proves part~(ii) for $z_1=z_2$, respectively.  The statements in that paper use one particular vector of prescribed power sums $\mathbf b$. The displayed estimates and their proofs depend only on $q$, $k$, and $m$, so they hold for every $\mathbf b$.

For distinct $i,\ell\in[m]$, a permutation of the $m$ coordinates sends the condition $z_i=z_\ell$ to $z_1=z_2$ and leaves each equation $\sum_{r=1}^m z_r^j=b_j$ unchanged.  Therefore Wan's estimate for $z_1=z_2$ proves part~(ii) for every pair $i\ne\ell$.

For part~(iii), fix $z_i=a$ and delete that coordinate.  The remaining $m-1$ variables satisfy $\sum_{r\ne i}z_r^j=b_j-a^j$ for $1\le j<k$. Since $k<m-1<q$, Wan's Proposition~5.2~\cite{wan2026np} applies to these $m-1$ variables and proves part~(iii).
\end{proof}

\section{Reed--Solomon lattices and ideal cosets}
\label{sec:descent}

Following the local-density approach in earlier $\SVP$ hardness reductions~\cite{Micciancio01,Micciancio12,BennettPeikert23}, we use an ideal lattice obtained from a Construction A Reed--Solomon lattice.  We require one fixed ideal coset containing many short ring elements represented by weight-$h$ binary coefficient vectors $\mathbf x\in\{0,1\}^{\F_q}$, with $[\mathbf x]=\sum_{a\in\F_q}x_a\zeta^a\in\mathcal{O}_K$.  The first step of the reduction stores an X3C selection vector $\boldsymbol\xi\in\{0,1\}^n$ in designated coordinates of $\mathbf x$. Using Wan's point-count estimates~\cite{wan2026np}, we show that for every $\boldsymbol\xi$ there are many such vectors $\mathbf x$ with $[\mathbf x]$ in the same fixed coset.  At the same time, the ideal is designed to have a minimum larger than the norms of these coset elements, so that its nonzero elements cannot create unintended short vectors.  This section constructs the ideal, proves the required count, and establishes the additional norm bounds for nonbinary representatives and scalar multiples.

\subsection{From lattice to ideal}

We first transfer an $\ell_2$-minimum bound for the Construction A Reed--Solomon lattice to a principal ideal of $\mathcal{O}_K$.

Fix an ordering $\F_q=\{a_1,\ldots,a_q\}$ and index the coordinates of $\Z^q$ accordingly.  For $1\le k<q$, define
\[
 \mathbf H_q(k):=
 \begin{pmatrix}
  1&1&\cdots&1\\
  a_1&a_2&\cdots&a_q\\
  a_1^2&a_2^2&\cdots&a_q^2\\
  \vdots&\vdots&\ddots&\vdots\\
  a_1^{k-1}&a_2^{k-1}&\cdots&a_q^{k-1}
 \end{pmatrix}
 =(a^j)_{0\le j<k,\ a\in\F_q}
 \in\F_q^{k\times q}.
\]
The matrix $\mathbf H_q(k)$ defines the lattice
\[
 \mathcal L_{q,k}:=\left\{\mathbf z\in\Z^{\F_q}:
 \mathbf H_q(k)\mathbf z=\mathbf0\text{ in }\F_q^k\right\},
\]
which is also known as the Construction A lattice of the Reed--Solomon code of dimension $q-k$ over $\F_q$.
\begin{samepage}
The following lemma gives a lower bound on the $\ell_2$-minimum of $\mathcal L_{q,k}$.
\begin{lemma}[{\cite[Theorem~14]{BennettPeikert23}}]
\label{lem:rs-minimum}
Let $q$ be prime and $1\le k\le q/2$.  Then the first minimum of $\mathcal L_{q,k}$ with respect to the $\ell_2$-norm $\|\cdot\|_2$ satisfies
\[
 \lambda_1(\mathcal L_{q,k})^2\ge 2k.
\]
\end{lemma}
\end{samepage}

Put $\mathfrak a_{q,k}:=\pi^k\mathcal{O}_K$ and write $\mathbf1_q^\perp:=\{\mathbf z\in\R^q: \mathbf1_q^{\mathsf T}\mathbf z=0\}$.  Let
\[
 \varphi:\Z[X]/(X^q-1)\longrightarrow \mathcal{O}_K,
 \qquad X\longmapsto\zeta,
\]
be the surjective ring homomorphism given by evaluation at $\zeta$.  Its kernel is generated by $\Phi_q(X)=1+X+\cdots+X^{q-1}$. We use the unique degree-$<q$ representative of each element $f\in\Z[X]/(X^q-1)$ and write $\bar f\in\F_q[X]$ for its reduction modulo $q$.

\begin{theorem}
\label{thm:cyclotomic-descent}
For $1\le k\le q/2$, the map $\mathbf z\mapsto[\mathbf z]$ restricts to an additive-group isomorphism
\[
 \mathcal L_{q,k}\cap\mathbf1_q^\perp
 \ \longrightarrow\ \mathfrak a_{q,k}.
\]
Consequently, with respect to the canonical norm,
\[
 \lambda_1(\mathfrak a_{q,k})^2\ge2kq.
\]
\end{theorem}

\begin{proof}
For $\mathbf z\in\Z^{\F_q}$, let $f:=\sum_{a\in\F_q}z_aX^a\in\Z[X]/(X^q-1)$, so that $\varphi(f)=[\mathbf z]$.  The coefficient of $Y^j$ in $f(1+Y)$ is $\sum_a z_a\binom aj$.  Thus
\[
  (X-1)^k\mid\bar f
  \quad\Longleftrightarrow\quad
  \sum_a z_a\binom aj=0\text{ in }\F_q\quad(0\le j<k).
\]
Since $k<q$, the functions $a\mapsto a^j$ and $a\mapsto\binom aj$, for $0\le j<k$, are related by an invertible triangular change of basis over $\F_q$.  Hence the defining equations of $\mathcal L_{q,k}$ are equivalent to $(X-1)^k\mid\bar f$. The identity $X^q-1=(X-1)^q$ in $\F_q[X]$ therefore gives
\[
 \mathbf z\in\mathcal L_{q,k}
 \quad\Longleftrightarrow\quad
 f\in\bigl(q,(X-1)^k\bigr)
 \triangleleft\Z[X]/(X^q-1).
\]
The image of $(q,(X-1)^k)$ under $\varphi$ is the ideal $(q,\pi^k)$.  Total ramification gives $(q)=(\pi^{q-1})$, so this image is $(\pi^k)=\mathfrak a_{q,k}$ since $k\le q-1$.

Furthermore, $\Phi_q(X)\equiv(X-1)^{q-1}\pmod q$, so $\ker(\varphi)=(\Phi_q)$ is contained in $(q,(X-1)^k)$.  Since $\varphi((q,(X-1)^k))=\mathfrak a_{q,k}$, it follows that
\begin{equation}
\label{eq:ideal-preimage}
 \varphi^{-1}(\mathfrak a_{q,k})
 =\bigl(q,(X-1)^k\bigr)
 \quad\text{in }\Z[X]/(X^q-1).
\end{equation}

To prove surjectivity of the restricted map, fix $v\in\mathfrak a_{q,k}$ and choose $f\in(q,(X-1)^k)$ with $\varphi(f)=v$. Since $f(1)\in q\Z$, replace $f$ by $f-(f(1)/q)\Phi_q$. The new $f$ remains in $(q,(X-1)^k)$, still satisfies $\varphi(f)=v$, and has coefficient sum zero.  Its coefficient vector $\mathbf z$ therefore lies in $\mathcal L_{q,k}\cap\mathbf1_q^\perp$ and satisfies $[\mathbf z]=v$.

For injectivity, if $\mathbf z\in\mathcal L_{q,k}\cap\mathbf1_q^\perp$ satisfies $[\mathbf z]=0$, then $\mathbf z\in\Z\mathbf1_q$ by the description of $\ker(\varphi)$.  Its coefficient sum is zero, so $\mathbf z=\mathbf0$.

Finally, for every nonzero $\mathbf z\in\mathcal L_{q,k}\cap\mathbf1_q^\perp$, \cref{lem:cyclotomic-form,lem:rs-minimum} give
\[
 \|[\mathbf z]\|^2=q\|\mathbf z\|_2^2\ge2kq.
\]
The isomorphism therefore gives the claimed minimum bound.
\end{proof}

\subsection{Binary representatives in ideal cosets}
\label{sec:binary-representatives}

For $\mathbf x\in\Z^{\F_q}$, its \emph{syndrome} $\mathbf H_q(k)\mathbf x\in\F_q^k$ has $j$th coordinate $\sum_{a\in\F_q}x_aa^j$ modulo $q$ for $0\le j<k$.  When $\mathbf x$ is binary, these coordinates are the power sums of the positions $a$ for which $x_a=1$.  \Cref{lem:coset-syndrome} identifies the vectors with fixed syndrome and coefficient sum with one coset of $\mathfrak a_{q,k}$.

\begin{lemma}
\label{lem:coset-syndrome}
Assume $1\le k\le q/2$. Fix $h\in\Z$. Let $\mathbf u=(u_0,\ldots,u_{k-1})\in\F_q^k$ satisfy $u_0=h\bmod q$. There exists $\mathbf z_0\in\Z^{\F_q}$ such that
\[
 \mathbf H_q(k)\mathbf z_0=\mathbf u,
 \qquad
 \mathbf1_q^{\mathsf T}\mathbf z_0=h.
\]
Moreover, evaluation at $\zeta$ induces a bijection
\[
 \{\mathbf z\in\Z^{\F_q}:
   \mathbf H_q(k)\mathbf z=\mathbf u,
   \ \mathbf1_q^{\mathsf T}\mathbf z=h\}
 \ \longrightarrow\
 [\mathbf z_0]+\mathfrak a_{q,k},
 \qquad \mathbf z\longmapsto[\mathbf z].
\]
Put $H:=h(q-h)$. For each $v\in[\mathbf z_0]+\mathfrak a_{q,k}$, let $\mathbf z$ be its unique representative with syndrome $\mathbf u$ and coefficient sum $h$. If $\mathbf z\in\{0,1\}^{\F_q}$, then it has weight $h$ and $\|v\|^2=H$. Otherwise, $\|v\|^2\ge H+2q$.
\end{lemma}

\begin{proof}
A Vandermonde $k\times k$ minor of $\mathbf H_q(k)$ has determinant $\prod_{1\le i<j\le k}(a_j-a_i)\ne0$ in $\F_q$.  Hence the matrix has rank $k$.  Choose $\mathbf x_0\in\Z^{\F_q}$ whose reduction modulo $q$ satisfies $\mathbf H_q(k)\mathbf x_0=\mathbf u$ over $\F_q$. Since $h-\mathbf1_q^{\mathsf T}\mathbf x_0\in q\Z$, adding this integer to the coordinate indexed by $0$ gives a vector $\mathbf z_0$ with coefficient sum $h$ and the same syndrome $\mathbf u$.

Every other vector with syndrome $\mathbf u$ and coefficient sum $h$ differs from $\mathbf z_0$ by an element of $\mathcal L_{q,k}\cap\mathbf1_q^\perp$.  The restricted isomorphism in \cref{thm:cyclotomic-descent} now gives the stated bijection.

By \cref{lem:cyclotomic-form},
\[
 \|v\|^2=q\sum_a z_a^2-h^2
 =H+q\sum_a z_a(z_a-1).
\]
For every integer $z_a$, the quantity $z_a(z_a-1)$ is a nonnegative even integer.  The sum is zero exactly when every $z_a\in\{0,1\}$, and is otherwise at least two.  This proves the stated claim.
\end{proof}

We next count binary vectors with a prescribed syndrome and fixed coordinate values. We use the following parameters to specialize \cref{prop:wan-power-sum-estimates}. Let $C_0$ and $k_0$ be sufficiently large absolute integer constants. For $k\ge k_0$ and an odd prime $q\ge(2k)^{C_0}$, put
\begin{equation}
 h:=\left\lfloor\frac{21k}{20}\right\rfloor.
\end{equation}
Set
\begin{equation}
\label{eq:HL-parameters}
 H:=h(q-h),\qquad L:=2kq.
\end{equation}
By \cref{lem:coset-syndrome}, every weight-$h$ binary vector $\mathbf x$ satisfies $\|[\mathbf x]\|^2=H$. By \cref{thm:cyclotomic-descent}, $\lambda_1(\mathfrak a_{q,k})^2\ge L$. Thus $H/L=21/40+o(1)<1$ as $k\to\infty$, uniformly for these $q$. In particular, the ring elements represented by weight-$h$ binary vectors are shorter than a shortest nonzero vector of $\mathfrak a_{q,k}$. The count below guarantees many such elements in every coset described in \cref{lem:coset-syndrome}.

For a prescribed syndrome $\mathbf u=(u_0,\ldots,u_{k-1})$, let $P\subseteq\F_q$ be the \emph{position set} indexing the coordinates fixed in advance, and let $\boldsymbol\eta\in\{0,1\}^P$ specify their values. Define
\[
 \mathcal F(\mathbf u;P,\boldsymbol\eta):=
 \{\mathbf x\in\{0,1\}^{\F_q}:\mathbf H_q(k)\mathbf x=\mathbf u,\
 \mathbf1_q^{\mathsf T}\mathbf x=h,\
 \mathbf x|_P=\boldsymbol\eta\}.
 \]
The first row of $\mathbf H_q(k)$ forces $u_0=h\bmod q$ in $\F_q$.

\begin{theorem}
\label{thm:uniform-completions}
For every $\mathbf u\in\F_q^k$ with $u_0=h\bmod q$, every position set $P\subseteq\F_q$ with $|P|\le\sqrt{k}$, and every $\boldsymbol\eta\in\{0,1\}^P$, put $t:=\sum_{a\in P}\eta_a$.  Then
\begin{equation}
 |\mathcal F(\mathbf u;P,\boldsymbol\eta)|
 =(1+o(1))\frac{q^{h-t-k+1}}{(h-t)!}.
\end{equation}
As $k\to\infty$, this estimate is uniform in $q$, $\mathbf u$, $P$, and $\boldsymbol\eta$. Its relative error is at most $q^{-\Omega(k)}+3k^2/q<1/4$.
\end{theorem}

\begin{proof}
Fix $\mathbf u$, $P$, and $\boldsymbol\eta$, and put $m_0:=h-t$.  By the parameter choices above,
\[
 m_0\ge21k/20-\sqrt{k}-1\ge k+2,
 \qquad m_0<q-2.
\]
The fixed coordinates contain exactly $t$ ones.  Subtract their contribution by defining
\begin{equation}
 u'_j:=u_j-\sum_{a\in P}\eta_aa^j
 \qquad(1\le j<k).
\end{equation}
For each $\mathbf x\in\mathcal F(\mathbf u;P,\boldsymbol\eta)$, ordering the positions of its $m_0$ ones outside $P$ gives exactly $m_0!$ tuples $(z_1,\ldots,z_{m_0})\in(\F_q\setminus P)^{m_0}$ with pairwise distinct entries satisfying
\begin{equation}
\label{eq:ordered-power-sums}
 \sum_{i=1}^{m_0}z_i^j=u'_j
 \qquad(1\le j<k).
\end{equation}
Conversely, each such tuple and the prescribed values on $P$ determine a unique vector in $\mathcal F(\mathbf u;P,\boldsymbol\eta)$, so the number of these tuples is $m_0!\,|\mathcal F(\mathbf u;P,\boldsymbol\eta)|$.

Let $A$ count the tuples satisfying \cref{eq:ordered-power-sums} without requiring their entries to be pairwise distinct or to lie outside $P$. In \cref{prop:wan-power-sum-estimates}, set $m=m_0$ and divide the error bound in part~(i) by $q^{m_0-k+1}$ and those in parts~(ii) and~(iii) by $q^{m_0-k}$.  The three ratios are at most
\[
 (2k)^{m_0}q^{-(m_0-k-2)/2}
 \le q^{-\Omega(k)}.
\]
This bound follows from $q\ge(2k)^{C_0}$, $m_0\ge21k/20-\sqrt{k}-1$, and the choices of $C_0$ and $k_0$. Part~(i) of \cref{prop:wan-power-sum-estimates} gives $|A/q^{m_0-k+1}-1|\le q^{-\Omega(k)}$. We bound the excluded tuples in two cases.
\par\noindent
\emph{Case I (collision).}  For fixed $i<\ell$, part~(ii) of \cref{prop:wan-power-sum-estimates} shows that the number of solutions with $z_i=z_\ell$ is at most $2q^{m_0-k}$.
\par\noindent
\emph{Case II (entry in $P$).}  For fixed $i$ and $a\in P$, part~(iii) of \cref{prop:wan-power-sum-estimates} bounds the number of solutions with $z_i=a$ by $2q^{m_0-k}$.

The union bound over $\binom{m_0}{2}$ equalities $z_i=z_\ell$ and $m_0|P|$ equalities $z_i=a$ with $a\in P$ shows that the nonnegative difference satisfies
\[
 0\le A-m_0!\,|\mathcal F(\mathbf u;P,\boldsymbol\eta)|
 \le2\biggl(\binom{m_0}{2}+m_0|P|\biggr)q^{m_0-k}
 \le3k^2q^{m_0-k}.
\]
Together with the estimate for $A$, this gives
\[
 \left|\frac{m_0!\,|\mathcal F(\mathbf u;P,\boldsymbol\eta)|}
 {q^{m_0-k+1}}-1\right|
 \le q^{-\Omega(k)}+\frac{3k^2}{q}<\frac14.
\]
Substituting $m_0=h-t$ proves the theorem.
\end{proof}

\subsection{Norm separation}

\Cref{thm:uniform-completions} supplies the binary representatives used by the checker in \cref{sec:x3c-checker}. \Cref{lem:coset-syndrome} gives the norm gap for nonbinary representatives. We now choose a coset center whose multiples by short nonzero elements of $\mathcal{O}_K$ remain separated from the ideal.

We first record the norm gap between roots of unity and all other nonzero elements of $\mathcal{O}_K$, together with a sparse representative for short elements. The sparse representative will be used in the polynomial argument for separating short scalar multiples from the ideal.
\begin{lemma}
\label{lem:sparse-tools}
The following statements hold.
\begin{enumerate}[label=(\roman*)]
\item Every $v\in\mathcal{O}_K$ has an integral coefficient representative $\mathbf z$ satisfying $\|\mathbf z\|_2^2\le2\|v\|^2/q$, and hence having at most $2\|v\|^2/q$ nonzero coefficients.
\item For every nonzero $w\in\mathcal{O}_K$, either $w=\pm\zeta^j$ for some $j$, in which case $\|w\|^2=q-1$, or~$\|w\|^2\ge2q-4$.
\end{enumerate}
\end{lemma}

\begin{proof}
For (i), let $\mathbf z=(z_a)_{a\in\F_q}\in\Z^{\F_q}$ satisfy $[\mathbf z]=v$. Choose $b\in\F_q$ minimizing $\sum_a(z_a-z_b)^2$. By \cref{lem:cyclotomic-form},
\[
 \sum_a(z_a-z_b)^2
 \le\frac1q\sum_{r\in\F_q}\sum_{a\in\F_q}(z_a-z_r)^2
 =\frac{2\|v\|^2}{q}.
\]
Subtracting $z_b$ from every coordinate gives another integral coefficient representative of $v$. Every nonzero integer coefficient has square at least one, so the resulting representative has at most $2\|v\|^2/q$ nonzero coefficients.

For (ii), if $\|w\|^2\ge q^2/4$, then $\|w\|^2\ge2q-4$.  Otherwise, part~(i) gives an integral coefficient representative with exactly $s<q/2$ nonzero coefficients.  If $s\ge2$, each of the $s(q-s)$ pairs formed by one nonzero coordinate and one zero coordinate contributes at least one to the pairwise norm formula in \cref{lem:cyclotomic-form}, so
\[
 \|w\|^2\ge s(q-s)\ge2(q-2).
\]
If $s=1$, then $w=t\zeta^j$ for a nonzero integer $t$.  The cases $t=\pm1$ are precisely the roots of unity, whereas $|t|\ge2$ gives squared norm at least $4(q-1)$.
\end{proof}

\begin{lemma}
\label{lem:sparse-polynomial}
If
\[
 F(X)=\sum_{i=1}^s c_iX^{a_i}\in\F_q[X]
\]
is nonzero, where every $c_i\ne0$ and the exponents $a_i\in\{0,\ldots,q-1\}$ are distinct, then its order of vanishing at $X=1$ is at most $s-1$.
\end{lemma}

\begin{proof}
Multiplicity at least $s$ would make the first $s$ coefficients of $F(1+Y)$ vanish. The resulting homogeneous system has determinant
\[
 \det\left(\binom{a_i}{j}\right)_{
       0\le j<s,\ 1\le i\le s}
 =\frac{\prod_{1\le i<\ell\le s}(a_\ell-a_i)}
        {\prod_{j=0}^{s-1}j!},
\]
which is nonzero in $\F_q$: the exponents are distinct, and $s\le q$ makes every factorial in the denominator nonzero.  Hence every $c_i$ would vanish. This contradicts $F\ne0$.
\end{proof}

We next choose $c$ so that $wc$ remains separated from $\mathfrak a_{q,k}$ for every nonzero $w\in\mathcal{O}_K$ with $\|w\|^2\le3q$.  Define $c\in\mathcal{O}_K$ by
\begin{equation}
\label{eq:robust-center}
 c:=h+\pi^{\lfloor k/2\rfloor}.
\end{equation}
Let $\mathbf c\in\Z^{\F_q}$ be the coefficient vector of $h+(1-X)^{\lfloor k/2\rfloor}$.  Then $[\mathbf c]=c$, and evaluating the polynomial at $X=1$ gives $\sum_{a\in\F_q}c_a=h$. Thus \cref{lem:coset-syndrome,thm:uniform-completions} identify and count binary representatives in $c+\mathfrak a_{q,k}$ with prescribed coordinate values.

The next theorem separates every short nonzero multiple of $c$ from $\mathfrak a_{q,k}$.
\begin{theorem}
\label{thm:robust-center}
For the parameter choices in \cref{sec:binary-representatives}, every nonzero $w\in\mathcal{O}_K$ with $\|w\|^2\le3q$ satisfies
\[
 \dist(wc,\mathfrak a_{q,k})^2>L/10.
\]
\end{theorem}

\begin{proof}
We prove the theorem by contradiction.  Suppose that there exist a nonzero $w\in\mathcal{O}_K$ with $\|w\|^2\le3q$ and an element $e\in wc+\mathfrak a_{q,k}$ such that $\|e\|^2\le L/10=kq/5$.  By part~(i) of \cref{lem:sparse-tools}, choose $f_e,f_w\in\Z[X]/(X^q-1)$ with $\varphi(f_e)=e$ and $\varphi(f_w)=w$, whose coefficient vectors have squared $\ell_2$-norm at most $2k/5$ and $6$, respectively.  In particular, $f_e$ and $f_w$ have at most $2k/5$ and $6$ nonzero coefficients. Put $m:=\lfloor k/2\rfloor$.

Since $\varphi(h+(1-X)^m)=c$ and $e-wc\in\mathfrak a_{q,k}$, \cref{eq:ideal-preimage} gives $f_e-f_w(h+(1-X)^m)\in(q,(X-1)^k)$. Reducing modulo $q$ and using $X^q-1=(X-1)^q$ in $\F_q[X]$ gives
\begin{equation}
\label{eq:center-congruence}
 \bar f_e-h\bar f_w
 \equiv\bar f_w(1-X)^m\pmod{(X-1)^k}.
\end{equation}
By \cref{eq:center-congruence}, $(X-1)^m$ divides $\bar f_e-h\bar f_w$.  This polynomial has at most $2k/5+6\le m$ nonzero coefficients by the choice of $k_0$, so \cref{lem:sparse-polynomial} gives $\bar f_e-h\bar f_w=0$.

Substituting this equality into \cref{eq:center-congruence} gives $(X-1)^{k-m}\mid\bar f_w$.  Since $\bar f_w$ has at most six nonzero coefficients and $k-m>5$, \cref{lem:sparse-polynomial} gives $\bar f_w=0$. Since $w\ne0$, the coefficient vector of $f_w$ is nonzero.  All its coefficients are divisible by $q$, so its squared $\ell_2$-norm is at least $q^2$.  This contradicts the upper bound of $6$ for the squared norm.
\end{proof}

\section{The X3C checker}
\label{sec:x3c-checker}

\Cref{sec:descent} supplies weight-$h$ binary representatives in fixed ideal cosets.  A representative $\mathbf x$ encodes the X3C selection vector $\boldsymbol\xi$ through prescribed coordinates.  This section constructs $U,V\in\mathcal O_K$ so that $\|U[\mathbf x]-V\|^2$ records the residual $\mathbf A\boldsymbol\xi-\mathbf1_M$, up to nonnegative error terms.  In \cref{sec:clean-representatives}, we define admissible representatives for which one of the error terms vanishes.  The final corollary records the conclusions used in \cref{sec:reduction}.

\subsection{Separated offsets}

Consider the X3C instance with universe $[M]$, three-element subsets $C_1,\ldots,C_n$, and incidence matrix $\mathbf A\in\{0,1\}^{M\times n}$.  A selection vector $\boldsymbol\xi\in\{0,1\}^n$ specifies which subsets are chosen. For each $i\in[M]$, define
\[
 \mathcal R_i:=\{j\in[n]:i\in C_j\}
 \qquad\text{and}\qquad
 s_i:=\sum_{j\in\mathcal R_i}\xi_j=(\mathbf A\boldsymbol\xi)_i.
\]
Thus $s_i$ is the number of selected sets containing $i$, and the sets selected by $\boldsymbol\xi$ form an exact cover exactly when $s_i=1$ for every $i$. The quantity encoded by the checker is
\begin{equation}
\label{eq:checker-target-decomposition}
 2\|\mathbf A\boldsymbol\xi-\mathbf1_M\|_2^2
 =
   \underbrace{2\sum_{i=1}^M(s_i^2-s_i)}_{\text{quadratic}}
   -\underbrace{2\sum_{i=1}^M s_i}_{\text{linear}}
    +2M.
\end{equation}
Since $\xi_j^2=\xi_j$, we have $s_i^2-s_i=\sum_{j,\ell\in\mathcal R_i,\ j\ne\ell}\xi_j\xi_\ell$, so the first sum in \cref{eq:checker-target-decomposition} is quadratic in $\boldsymbol\xi$.

Define the \emph{incidence set}
\[
 \mathcal I:=\{(i,j)\in[M]\times[n]:j\in\mathcal R_i\}.
\]
Every set $C_j$ contains three elements. Under the assumptions in \cref{sec:x3c-and-sidon},
\begin{equation}
\label{eq:incidence-count}
 W:=|\mathcal I|=3n,
 \qquad M\le W.
\end{equation}

Let $0\le\alpha_1<\cdots<\alpha_n=\alpha_\star$ be the Sidon positions from \cref{lem:sidon}.  We encode $\boldsymbol\xi$ in a vector $\mathbf x\in\{0,1\}^{\F_q}$ by requiring $x_{\alpha_j}=\xi_j$ for $j\in[n]$, so the coordinate indexed by $\alpha_j$ stores the selection bit $\xi_j$ for $C_j$.  We call $\alpha_j$ a \emph{witness position}.  For $a,b\in\F_q$, the \emph{offset} from $a$ to $b$ is the cyclic difference $b-a\in\F_q$.  We separate the offsets associated with different rows of $\mathbf A$.  Define
\begin{equation}
\label{eq:row-shifts}
 \tau:=2\alpha_\star+1,
 \qquad
 \beta_i:=\tau i\quad(i\in[M]),
 \qquad
  d_{ij}:=\beta_i-\alpha_j\quad((i,j)\in\mathcal I).
\end{equation}
We regard $\alpha_j$, $\beta_i$, and $d_{ij}$ as integers and reduce them modulo $q$ when using them as coordinate indices or offsets. Here $\beta_i$ is the position assigned to row $i$ of $\mathbf A$, and $d_{ij}$ is the offset from the witness position $\alpha_j$ to $\beta_i$.  We choose $q$ large enough that
\begin{equation}
\label{eq:row-shift-q-bound}
 q>4\beta_M=4\tau M.
\end{equation}

We use two types of offsets.  The \emph{witness offsets} are the cyclic differences between distinct witness positions:
\[
\mathcal D_{\mathrm{wit}}
 :=\{\alpha_\ell-\alpha_j:j,\ell\in[n],\ j\ne\ell\}.
\]
By the Sidon property and $q>4\beta_M>2\alpha_\star$, these differences remain distinct modulo $q$. The \emph{bad offsets} are the cross-row differences and the sums of the offsets $d_{ij}$:
\[
 \mathcal D_{\mathrm{bad}}
 :=\{d_{ij}-d_{r\ell}:(i,j),(r,\ell)\in\mathcal I,\ i\ne r\}
 \cup
 \{d_{ij}+d_{r\ell}:(i,j),(r,\ell)\in\mathcal I\}.
\]

The expansion of $\|U[\mathbf x]-V\|^2$ in the proof of \cref{thm:checker-score} uses witness offsets to represent the products $\xi_j\xi_\ell$ in the quadratic term of \cref{eq:checker-target-decomposition}.  For $d\in\F_q$, the \emph{correlation} of $\mathbf x$ at offset $d$ counts all ordered pairs $(a,a+d)$ with $x_a=x_{a+d}=1$, as defined in \cref{sec:checker-construction}.  The position pair $(\alpha_j,\alpha_\ell)$ contributes $\xi_j\xi_\ell=x_{\alpha_j}x_{\alpha_\ell}$ to the correlation of $\mathbf x$ at $d=\alpha_\ell-\alpha_j$.  Correlations at bad offsets contribute to a nonnegative error term.  To preserve the products at witness offsets while eliminating the correlations at bad offsets, the two offset sets are designed to be disjoint.  The next lemma proves this separation and bounds the number of bad offsets.
\begin{lemma}
\label{lem:row-shift-separation}
Under the bound in \cref{eq:row-shift-q-bound}, the $2W$ residues $d_{ij}$ and $-d_{ij}$, for $(i,j)\in\mathcal I$, are nonzero and pairwise distinct. Moreover,
\[
 0\notin\mathcal D_{\mathrm{bad}},
 \qquad
 \mathcal D_{\mathrm{bad}}\cap\mathcal D_{\mathrm{wit}}=\varnothing,
 \qquad
 |\mathcal D_{\mathrm{bad}}|=O(W^2)=O(n^2).
\]
\end{lemma}

\begin{proof}
For every $(i,j)\in\mathcal I$,
\[
 \alpha_\star+1=\beta_1-\alpha_\star
 \le d_{ij}\le\beta_M<q/4,
\]
so $d_{ij}$ is nonzero modulo $q$.  Fix distinct $(i,j),(r,\ell)\in\mathcal I$.

\smallskip
\noindent\emph{Case 1: $i=r$.}
\[
 d_{ij}-d_{i\ell}=\alpha_\ell-\alpha_j\ne0,
 \qquad
 |d_{ij}-d_{i\ell}|\le\alpha_\star.
\]

\smallskip
\noindent\emph{Case 2: $i\ne r$.}
\[
 d_{ij}-d_{r\ell}
 =\beta_i-\beta_r+\alpha_\ell-\alpha_j
 =\tau(i-r)+\alpha_\ell-\alpha_j.
\]
Consequently,
\[
 \alpha_\star+1\le|d_{ij}-d_{r\ell}|
 \le \tau(M-1)+\alpha_\star=\beta_M-\alpha_\star-1.
\]
For all $(i,j),(r,\ell)\in\mathcal I$, one also has
\[
 2\alpha_\star+2=2\beta_1-2\alpha_\star
 \le d_{ij}+d_{r\ell}\le2\beta_M.
\]
Since $q>4\beta_M$, the bounds above give $0<|d_{ij}-d_{r\ell}|<q/2$ whenever $(i,j)\ne(r,\ell)$, so $d_{ij}\not\equiv d_{r\ell}\pmod q$. For all $(i,j),(r,\ell)\in\mathcal I$, we also have $0<d_{ij}+d_{r\ell}<q/2$, so $d_{ij}\not\equiv-d_{r\ell}\pmod q$. Thus the $2W$ residues $d_{ij}$ and $-d_{ij}$ are nonzero and pairwise distinct.

Every witness offset is represented in $[-\alpha_\star,\alpha_\star]$.  By contrast, the integer differences $d_{ij}-d_{r\ell}$ with $i\ne r$ and the integer sums $d_{ij}+d_{r\ell}$ lie in $(-q/2,q/2)$, with absolute value greater than $\alpha_\star$.  Hence $0\notin\mathcal D_{\mathrm{bad}}$ and $\mathcal D_{\mathrm{bad}}\cap\mathcal D_{\mathrm{wit}}=\varnothing$. The two sets in the definition of $\mathcal D_{\mathrm{bad}}$ contain at most $O(W^2)=O(n^2)$ offsets.
\end{proof}

\subsection{Checker elements and norm formula}
\label{sec:checker-construction}

Fix $\mathbf x\in\{0,1\}^{\F_q}$ of weight $h$ and define $\xi_j:=x_{\alpha_j}$.  Recall $s_i:=\sum_{j\in\mathcal R_i}\xi_j$.  By \cref{lem:cyclotomic-form} and \cref{eq:HL-parameters}, $\|[\mathbf x]\|^2=H=h(q-h)$.

For $d\in\F_q$, define the \emph{cyclic correlation} of $\mathbf x$ at offset $d$ by
\[
 \mathcal{C}_{\mathbf x}(d):=\sum_{a\in\F_q}x_ax_{a+d}.
\]
The indices lie in $\F_q$, and $a+d$ is computed modulo $q$. Thus $\mathcal C_{\mathbf x}(d)$ counts the ordered pairs $(a,a+d)\in\F_q^2$ for which $x_a=x_{a+d}=1$.

\Cref{eq:checker-target-decomposition} separates the target into quadratic and linear parts.  We first define $g\in \mathcal{O}_K$ so that $\|g[\mathbf x]\|^2$ represents the quadratic part through cyclic correlations.  Next, $V$ encodes the linear part, and the checker combines both parts in $\|U[\mathbf x]-V\|^2$.
\paragraph{Quadratic term.}
Using the separated offsets $d_{ij}$, define
\[
 g:=\sum_{(i,j)\in\mathcal I}
       \bigl(\zeta^{d_{ij}}+\zeta^{-d_{ij}}\bigr)\in \mathcal{O}_K.
\]
We have $g=\overline g$.  By \cref{lem:row-shift-separation}, the $2W$ exponents in the sum defining $g$ are distinct modulo $q$.

\paragraph{Linear term.}
Define
\[
 V:=6\sum_{j=1}^n\zeta^{\alpha_j}\in \mathcal{O}_K.
\]
The coefficient $6=2\cdot3$ comes from the factor $2$ in \cref{eq:checker-target-decomposition} and the three elements in each selected set.

\paragraph{Checker elements.}
To obtain the remainder bound in \cref{thm:checker-score}, put $Q:=128W^4h$ and define
\[
 U:=Q+g^2\in\mathcal{O}_K.
\]
Since $g=\overline g$, we have $\langle[\mathbf x],g^2[\mathbf x]\rangle=\|g[\mathbf x]\|^2$. Thus the expansion of $\|U[\mathbf x]-V\|^2$ contains $2Q\|g[\mathbf x]\|^2$, which represents the quadratic part. The following lemma gives norm bounds for multiplication by $U$ and $V$.

\begin{lemma}
\label{lem:checker-U-properties}
For every $y,w\in K$, one has $\|Uy\|\ge Q\|y\|$ and $\|Vw\|\le2W\|w\|$.
\end{lemma}

\begin{proof}
For every embedding $\sigma:K\hookrightarrow\C$, $\sigma(g)$ is real, so $\sigma(U)=Q+\sigma(g)^2\ge Q>0$. Hence $|\sigma(Uy)|\ge Q|\sigma(y)|$. The triangle inequality gives $|\sigma(g)|\le2W$ and $|\sigma(V)|\le2W$, so $|\sigma(Vw)|\le2W|\sigma(w)|$.  Summing squares over all embeddings proves the two inequalities.
\end{proof}

\paragraph{Norm formula.}
We call $U[\mathbf x]-V\in\mathcal O_K$ the \emph{checker value}. Put $\kappa:=2q/Q$ and $\mathcal R(\mathbf x):=\|g^2[\mathbf x]-V\|^2/Q^2$. The next theorem relates the squared norm of the checker value to $\|\mathbf A\boldsymbol\xi-\mathbf1_M\|_2^2$ and bounds the remainder $\mathcal R(\mathbf x)$.

\begin{theorem}
\label{thm:checker-score}
For every $\mathbf x\in\{0,1\}^{\F_q}$ of weight $h$, define $\boldsymbol\xi\in\{0,1\}^n$ by $\xi_j:=x_{\alpha_j}$ for $j\in[n]$.  Then there exists $\mathcal E(\mathbf x)\ge0$ such that
\begin{equation}
 \frac{\|U[\mathbf x]-V\|^2}{Q^2}
 =B_0+2\kappa\|\mathbf A\boldsymbol\xi-\mathbf1_M\|_2^2
   +\kappa\mathcal E(\mathbf x)+\mathcal R(\mathbf x),
\end{equation}
where
\begin{equation}
\label{eq:checker-baseline}
 B_0:=H+\kappa\biggl(2Wh-2M
       +\frac{2hW-4h^2W^2}{q}\biggr)\in\mathbb{Q}
\end{equation}
is independent of $\mathbf x$. Moreover, $0\le\mathcal R(\mathbf x)<\kappa/8$.
\end{theorem}

\begin{proof}
\noindent\textbf{Quadratic term.}
For $a\in\F_q$, define
\[
 y_a:=\sum_{(i,j)\in\mathcal I}
      (x_{a-d_{ij}}+x_{a+d_{ij}}).
\]
Then $g[\mathbf x]=\sum_a y_a\zeta^a$. To compute $\|g[\mathbf x]\|^2$ by \cref{lem:cyclotomic-form}, we determine $\sum_a y_a$ and $\sum_a y_a^2$. For each $(i,j)\in\mathcal I$, the vectors $(x_{a-d_{ij}})_{a\in\F_q}$ and $(x_{a+d_{ij}})_{a\in\F_q}$ are coordinate permutations of $\mathbf x$, so both have coordinate sum $h$. Hence $\sum_a y_a=2Wh$.  For $r,s\in\F_q$, cyclic reindexing gives
\[
 \sum_{a\in\F_q}x_{a-r}x_{a-s}
 =\mathcal{C}_{\mathbf x}(r-s),
 \qquad
 \mathcal{C}_{\mathbf x}(-r)=\mathcal{C}_{\mathbf x}(r).
\]
Expanding $\sum_a y_a^2$ and using the preceding identities gives
\[
\sum_{a\in\F_q}y_a^2
 =2Wh+2\sum_{\substack{(i,j),(t,\ell)\in\mathcal I\\(i,j)\ne(t,\ell)}}
   \mathcal{C}_{\mathbf x}(d_{ij}-d_{t\ell})
 +2\sum_{(i,j),(t,\ell)\in\mathcal I}
   \mathcal{C}_{\mathbf x}(d_{ij}+d_{t\ell}).
\]
The squared terms $x_{a-d_{ij}}^2$ and $x_{a+d_{ij}}^2$ contribute $2Wh$ after summing over $a\in\F_q$ and $(i,j)\in\mathcal I$. Both double sums are over ordered pairs in $\mathcal I\times\mathcal I$. The products $x_{a-d_{ij}}x_{a-d_{t\ell}}$ and $x_{a+d_{ij}}x_{a+d_{t\ell}}$ each contribute $\mathcal C_{\mathbf x}(d_{ij}-d_{t\ell})$ after summing over $a$. The sum over $a$ of $x_{a-d_{ij}}x_{a+d_{t\ell}}+x_{a+d_{ij}}x_{a-d_{t\ell}}$ is $2\mathcal C_{\mathbf x}(d_{ij}+d_{t\ell})$. For the terms with $i=t$ and $j\ne\ell$, we have $d_{ij}-d_{i\ell}=\alpha_\ell-\alpha_j$. The ordered pair $(\alpha_j,\alpha_\ell)$ contributes $\xi_j\xi_\ell$ to $\mathcal C_{\mathbf x}(\alpha_\ell-\alpha_j)$. The difference $\mathcal C_{\mathbf x}(\alpha_\ell-\alpha_j)-\xi_j\xi_\ell$ counts all other ordered pairs $(a,a+\alpha_\ell-\alpha_j)$ with $x_a=x_{a+\alpha_\ell-\alpha_j}=1$, so it is nonnegative. Since $\sum_{\substack{j,\ell\in\mathcal R_i\\j\ne\ell}} \xi_j\xi_\ell=s_i^2-s_i$, the products $\xi_j\xi_\ell$ contribute $2\sum_i(s_i^2-s_i)$.

We collect the remaining terms in
\begin{equation}
\label{eq:checker-additional-correlations}
\mathcal E(\mathbf x):=
  2\sum_{i=1}^M\sum_{\substack{j,\ell\in\mathcal R_i\\j\ne\ell}}
  \bigl(
   \mathcal{C}_{\mathbf x}(\alpha_\ell-\alpha_j)-\xi_j\xi_\ell
  \bigr)
  +2\sum_{\substack{(i,j),(t,\ell)\in\mathcal I\\i\ne t}}
    \mathcal{C}_{\mathbf x}(d_{ij}-d_{t\ell})
  +2\sum_{(i,j),(t,\ell)\in\mathcal I}
    \mathcal{C}_{\mathbf x}(d_{ij}+d_{t\ell}).
\end{equation}
The differences in the first term and the correlation counts in the other two terms are nonnegative.  Hence $\mathcal E(\mathbf x)\ge0$, and the expansion becomes
\[
 \sum_a y_a^2
 =2Wh+2\sum_{i=1}^M(s_i^2-s_i)+\mathcal E(\mathbf x).
\]

Using \cref{lem:cyclotomic-form} and $\sum_a y_a=2Wh$, we obtain
\[
 \|g[\mathbf x]\|^2
 =q\left(
 2Wh+2\sum_i(s_i^2-s_i)+\mathcal E(\mathbf x)
 \right)-4W^2h^2.
\]

\smallskip
\noindent\textbf{Linear term.}
In the expression $V=6\sum_{j=1}^n\zeta^{\alpha_j}$, the coefficient sum is $6n=2W$. Since each selected set contains three universe elements, $3\sum_j\xi_j=\sum_i s_i$.  Hence \cref{lem:cyclotomic-form} gives
\[
 \langle [\mathbf x],V\rangle
 =6q\sum_j\xi_j-2hW
 =2q\sum_i s_i-2hW.
\]

\begin{samepage}
\smallskip
\noindent\textbf{Combining the terms.}
Since $g=\overline g$, one has $\langle[\mathbf x],g^2[\mathbf x]\rangle=\|g[\mathbf x]\|^2$. The definitions of $U$ and $V$, the equality $\|[\mathbf x]\|^2=H$, and the two preceding identities give
\[
\begin{aligned}
 \frac{\|U[\mathbf x]-V\|^2}{Q^2}
 &=H+\frac2Q\|g[\mathbf x]\|^2
   -\frac2Q\langle [\mathbf x],V\rangle+\mathcal R(\mathbf x)\\
 &=B_0+2\kappa\sum_i(s_i-1)^2+\kappa\mathcal E(\mathbf x)+\mathcal R(\mathbf x)\\
 &=B_0+2\kappa\|\mathbf A\boldsymbol\xi-\mathbf1_M\|_2^2
   +\kappa\mathcal E(\mathbf x)+\mathcal R(\mathbf x).
\end{aligned}
\]
\end{samepage}
The last two equalities use $\sum_i(s_i^2-2s_i)=\|\mathbf A\boldsymbol\xi-\mathbf1_M\|_2^2-M$. The remaining terms independent of $\mathbf x$ give $B_0$ from \cref{eq:checker-baseline}.

Finally, $H\le hq$ and the embedding bounds in the proof of \cref{lem:checker-U-properties} give
\[
 0\le\mathcal R(\mathbf x)
 \le\frac{\bigl(4W^2\sqrt{hq}+2W\sqrt q\bigr)^2}{Q^2}
 <\frac\kappa8.
\]
The last inequality uses $W\ge3$, $h\ge1$, and the definition of $Q$.
\end{proof}

\subsection{Representatives with prescribed correlations}
\label{sec:clean-representatives}

Let $\mathbf x\in\{0,1\}^{\F_q}$ have weight $h$, and put $\xi_j:=x_{\alpha_j}$ for every $j\in[n]$. We call $\mathbf x$ an \emph{admissible representative} if
\begin{equation}
\label{eq:enhanced-clean-correlations}
 \mathcal{C}_{\mathbf x}(\alpha_\ell-\alpha_j)=\xi_j\xi_\ell
 \quad(j,\ell\in[n],\ j\ne\ell),
 \qquad
 \mathcal{C}_{\mathbf x}(d)=0
 \quad(d\in\mathcal D_{\mathrm{bad}}).
\end{equation}
The first equality allows only the ordered pair $(\alpha_j,\alpha_\ell)$ to contribute at the witness offset $\alpha_\ell-\alpha_j$.  This ordered pair contributes $\xi_j\xi_\ell$. The second equality excludes every ordered pair $(a,a+d)$ with $d\in\mathcal D_{\mathrm{bad}}$ and $x_a=x_{a+d}=1$.  Hence every term in \cref{eq:checker-additional-correlations} vanishes and $\mathcal E(\mathbf x)=0$ whenever $\mathbf{x}$ is an admissible representative.

The next theorem uses \cref{thm:uniform-completions} and a union bound to prove that an admissible representative always exists.
\begin{theorem}
\label{thm:enhanced-clean}
With the parameters from \cref{sec:binary-representatives}, assume that $n+2\le\lfloor\sqrt{k}\rfloor$ and \mbox{\cref{eq:row-shift-q-bound}} holds.  For every syndrome $\mathbf u\in\F_q^k$ with $u_0=h\bmod q$, and every $\boldsymbol\xi\in\{0,1\}^n$, there is an admissible representative $\mathbf x\in\{0,1\}^{\F_q}$ such that
\[
 \mathbf H_q(k)\mathbf x=\mathbf u,
 \qquad
 \sum_a x_a=h,
 \qquad
 x_{\alpha_j}=\xi_j\quad(j\in[n]).
\]
\end{theorem}

\begin{proof}
Put $P:=\{\alpha_1,\ldots,\alpha_n\}$. Let $\mathcal F$ be the set of weight-$h$ vectors $\mathbf x\in\{0,1\}^{\F_q}$ with syndrome $\mathbf u$ and $x_{\alpha_j}=\xi_j$ for every $j\in[n]$. Applying \cref{thm:uniform-completions} with position set $P$ and values $\xi_j$ at $\alpha_j$ gives $|\mathcal F|>0$.

By \cref{lem:sidon,eq:row-shift-q-bound}, the ordered pair $(\alpha_j,\alpha_\ell)$ is the unique ordered pair in $P$ at the witness offset $d=\alpha_\ell-\alpha_j$.  The contribution $x_{\alpha_j}x_{\alpha_\ell}=\xi_j\xi_\ell$ is already fixed. The first equality in \cref{eq:enhanced-clean-correlations} therefore requires $x_ax_{a+d}=0$ only for $a\ne\alpha_j$. For $d\in\mathcal D_{\mathrm{bad}}$, the second equality requires $x_ax_{a+d}=0$ for every $a\in\F_q$.  By \cref{lem:row-shift-separation}, no ordered pair in $P$ has a bad offset. Thus every pair whose product is required to vanish has at least one coordinate outside $P$.  All these offsets are nonzero.

For a set $J\subseteq\F_q\setminus P$ of one or two additional coordinates, put
\[
 \mathcal F_J:=\{\mathbf x\in\mathcal F:x_a=1
 \text{ for every }a\in J\}.
\]
Since $|P|+|J|\le n+2\le\sqrt{k}$, \cref{thm:uniform-completions} applies to both $\mathcal F_J$ and $\mathcal F$.  Let $t$ be the weight of $\boldsymbol\xi$, so $t=\sum_{j=1}^n\xi_j$. By \cref{thm:uniform-completions}, each count lies between $3/4$ and $5/4$ times its main term.  Dividing the upper bound for $|\mathcal F_J|$ by the lower bound for $|\mathcal F|$ gives
\[
 \frac{|\mathcal F_J|}{|\mathcal F|}
  \le\frac53\frac{(h-t)!}{q^{|J|}(h-t-|J|)!}
  <2(h/q)^{|J|}.
\]

Consider one of the conditions $x_ax_{a+d}=0$ above. If exactly one of $a,a+d$ lies outside $P$, no representative violates the condition when the coordinate in $P$ is fixed to zero.  When the coordinate in $P$ is fixed to one, the violating representatives form $\mathcal F_J$ for $J=\{a,a+d\}\setminus P$, so $|\mathcal F_J|/|\mathcal F|\le2h/q$.  If both $a$ and $a+d$ lie outside $P$, the violating representatives form $\mathcal F_{\{a,a+d\}}$ and occupy a fraction of at most $2(h/q)^2$.

The definition of $\mathcal D_{\mathrm{wit}}$ and \cref{lem:row-shift-separation} give $|\mathcal D_{\mathrm{wit}}\cup\mathcal D_{\mathrm{bad}}| =O(n^2)$.  For each of these offsets, at most $2n$ ordered pairs $(a,a+d)$ have exactly one of $a,a+d$ outside $P$, and at most $q$ have both $a,a+d$ outside $P$. A union bound shows that the proportion of representatives violating at least one of these conditions is at most
\[
 O(n^2)\left(2n\cdot\frac{2h}{q}
       +q\cdot2\left(\frac hq\right)^2\right)
 =O\!\left(\frac{n^3h+n^2h^2}{q}\right).
\]
Since $n\le\sqrt{k}$ and $h=O(k)$, the proportion of violating representatives is $O(k^3/q)$, uniformly in $\mathbf u$ and $\boldsymbol\xi$. By $q\ge(2k)^{C_0}$ and the choice of $k_0$, this proportion is less than one. Therefore some representative in $\mathcal F$ is admissible.
\end{proof}

We summarize the checker bounds used in \cref{sec:reduction}.
\begin{corollary}
\label{cor:checker-interface}
Under the assumptions of \cref{thm:enhanced-clean}, let $c\in\mathcal O_K$ have an integral coefficient representative $\mathbf c\in\Z^{\F_q}$ with $\sum_a c_a=h$. Then for every $\boldsymbol\xi\in\{0,1\}^n$, there exists an admissible representative $\mathbf x\in\{0,1\}^{\F_q}$ of weight $h$ such that $[\mathbf x]\in c+\mathfrak a_{q,k}$, $x_{\alpha_j}=\xi_j$ for $j\in[n]$, and
\begin{equation}
\label{eq:clean-checker-interface}
 B_0+2\kappa\|\mathbf A\boldsymbol\xi-\mathbf1_M\|_2^2
  \le\frac{\|U[\mathbf x]-V\|^2}{Q^2}
  \le B_0+2\kappa\|\mathbf A\boldsymbol\xi-\mathbf1_M\|_2^2+\frac\kappa8.
\end{equation}
\end{corollary}

\begin{proof}
Put $\mathbf u_c:=\mathbf H_q(k)\mathbf c\in\F_q^k$ and choose $\mathbf x$ by \cref{thm:enhanced-clean} with $\mathbf u=\mathbf u_c$. The vectors $\mathbf x$ and $\mathbf c$ have syndrome $\mathbf u_c$ and coefficient sum $h$.  Hence \cref{lem:coset-syndrome} gives $[\mathbf x]\in c+\mathfrak a_{q,k}$. Admissibility gives $\mathcal E(\mathbf x)=0$, so \cref{thm:checker-score} gives the norm bounds in \cref{eq:clean-checker-interface}.
\end{proof}

\section{The rank-two reduction}
\label{sec:reduction}

Let $q$ be an odd prime with the parameters specified in \cref{sec:binary-representatives}. This section assembles the checker and the ideal coset into a rank-two module and chooses the scale $\Gamma$ for its second coordinate.  Let $U,V\in\mathcal O_K$ and $Q$ be the checker elements and scale constructed in \cref{sec:checker-construction} for the given X3C instance, respectively.  The reduction uses the ideal $\mathfrak a:=\mathfrak a_{q,k}=\pi^k\mathcal{O}_K$ and the center $c\in\mathcal{O}_K$ from \cref{eq:robust-center}. Assume that $n+2\le\lfloor\sqrt{k}\rfloor$ and \cref{eq:row-shift-q-bound} holds. The center $c$ fixes the ideal coset whose representatives encode X3C selections.  The elements $U$ and $V$ test the X3C equations in the first coordinate.  The integer $\Gamma$ weights the second coordinate. Define the generators and their $\mathcal{O}_K$-span by
\[
\begin{gathered}
 \mathbf m_1:=(U\pi^k,0),
 \qquad
 \mathbf m_2:=(Uc-V,-Q\Gamma),
 \\
 \mathcal M:=\mathcal{O}_K\mathbf m_1
 +\mathcal{O}_K\mathbf m_2\subset\mathcal O_K^2.
\end{gathered}
\]
Taking $y=1$ in \cref{lem:checker-U-properties} gives $U\ne0$, and the choice of $\Gamma$ below is positive. The generator matrix has nonzero determinant $-Q\Gamma U\pi^k$, so $\mathcal M$ is a free module of rank two and $\Z$-rank $2(q-1)$. Hence every vector in $\mathcal M$ has the unique form
\begin{equation}
\label{eq:module-vector-exact}
 \mathbf v(y,w)
 :=\bigl(Uy-Vw,-Q\Gamma w\bigr),
 \qquad y,w\in\mathcal{O}_K,\quad y-wc\in\mathfrak a.
\end{equation}
The coefficients of $\mathbf v(y,w)$ in the module basis $\mathbf m_1,\mathbf m_2$ are $(y-wc)/\pi^k$ and $w$, respectively. When $w=1$ and $y=[\mathbf x]$ for some $\mathbf x\in\{0,1\}^{\F_q}$ of weight $h$, the first coordinate is the checker value $U[\mathbf x]-V$.  \Cref{cor:checker-interface} guarantees the existence of such an admissible representative for the completeness proof with $w=1$.  \Cref{thm:checker-score} gives a lower bound for every weight-$h$ binary coefficient vector.  The second coordinate contributes $Q^2\Gamma^2\|w\|^2$ to the squared norm.  For soundness, we consider separately $w=0$, roots of unity, and elements that are not roots of unity.

\subsection{Parameters}

We choose $\Gamma$ and $S$ for the norm comparisons in the proof of \cref{thm:exact-X3C-encoding}.  Recall $H=h(q-h)$, $L=2kq$, $W=3n$, and $Q=128W^4h$, and let $B_0$ be as in \cref{eq:checker-baseline}.  Define
\begin{equation}
\label{eq:F-S-nu-exact}
 \Gamma:=\lfloor\sqrt{9k/10}\rfloor,\quad
 F:=\Gamma^2(q-1),\quad
 S:=Q^2(B_0+F+\kappa),\quad
 \nu:=2W/(Q\Gamma).
\end{equation}
When $w$ is a root of unity, the second coordinate contributes $Q^2F$ to $\|\mathbf v(y,w)\|^2$. By \cref{eq:checker-baseline}, both $Q^2B_0$ and $Q^2\kappa=2qQ$ are integers, so $S\in\Z$.

The formula in \cref{eq:module-vector-exact} also defines $\mathbf v(y,w)$ for arbitrary $y,w\in K$.

\begin{lemma}
\label{lem:gap-and-stability}
For the parameters in \cref{eq:F-S-nu-exact}, we have $0<\nu<1$. For every $y,w\in K$,
\begin{equation}
\label{eq:global-stability-exact}
 \|\mathbf v(y,w)\|^2
 \ge Q^2(1-\nu)\bigl(\|y\|^2+\Gamma^2\|w\|^2\bigr).
\end{equation}
\end{lemma}

\begin{proof}
Since $Q>2W$ and $\Gamma\ge1$, we have $0<\nu<1$. \Cref{lem:checker-U-properties} gives $\|Uy\|\ge Q\|y\|$ and $\|Vw\|\le2W\|w\|=\nu Q\Gamma\|w\|$.  Hence
\[
 2|\langle Uy,Vw\rangle|
 \le2\nu Q\Gamma\|Uy\|\|w\|
 \le\nu\bigl(\|Uy\|^2+Q^2\Gamma^2\|w\|^2\bigr).
\]
Expanding $\|\mathbf v(y,w)\|^2$ and using this bound together with $\|Uy\|\ge Q\|y\|$ proves \cref{eq:global-stability-exact}.
\end{proof}

\subsection{Completeness and soundness}

\begin{theorem}
\label{thm:exact-X3C-encoding}
With the parameters from \cref{sec:binary-representatives,eq:F-S-nu-exact}, assume $n+2\le\lfloor\sqrt{k}\rfloor$ and \cref{eq:row-shift-q-bound} holds. Then the following implications hold:
\[
\begin{aligned}
 \exists\boldsymbol\xi\in\{0,1\}^n:
 \mathbf A\boldsymbol\xi=\mathbf1_M
 &\quad\Longrightarrow\quad
 \lambda_1(\mathcal M)^2\mathrel{<}S,\\
 \forall\boldsymbol\xi\in\{0,1\}^n:
 \mathbf A\boldsymbol\xi\ne\mathbf1_M
 &\quad\Longrightarrow\quad
 \lambda_1(\mathcal M)^2>S.
\end{aligned}
\]
\end{theorem}

\begin{proof}
\noindent\textbf{Completeness.}
Let $\boldsymbol\xi\in\{0,1\}^n$ satisfy $\mathbf A\boldsymbol\xi=\mathbf1_M$. \Cref{cor:checker-interface} gives an admissible representative $\mathbf x\in\{0,1\}^{\F_q}$ such that $[\mathbf x]\in c+\mathfrak a$ and $\|U[\mathbf x]-V\|^2\le Q^2(B_0+\kappa/8)$. By \cref{eq:module-vector-exact}, the vector $\mathbf v([\mathbf x],1)\in\mathcal M$ has second coordinate $-Q\Gamma\ne0$ and satisfies
\[
 \|\mathbf v([\mathbf x],1)\|^2
 \le Q^2\biggl(B_0+F+\frac\kappa8\biggr)
 =S-\frac{7\kappa Q^2}{8}<S.
\]

\medskip
\noindent\textbf{Soundness.}
We first record the parameter estimates used in the soundness proof. Since $h=\Theta(k)$ and $W\ge3$, the definitions of $Q$, $\Gamma$, $\nu$, and $\kappa$ give
\[
 \Gamma^2=\frac{9k}{10}+O(\sqrt{k}),\qquad
 \nu=O(k^{-3/2}),\qquad \kappa=O(q/k).
\]
By \cref{eq:checker-baseline} and $M\le W$, we have $|B_0-H|=O(q/W^3+h/W^2)=o(L)$.  Since $h=21k/20+O(1)$ and $k=o(q)$,
\[
 H=\frac{21}{40}L+o(L),\qquad
 F=\frac{18}{40}L+o(L),\qquad
 \frac{S}{Q^2}=\frac{39}{40}L+o(L).
\]
These estimates hold as $k\to\infty$, uniformly for $q\ge(2k)^{C_0}$, $1\le M\le3n$, and $n+2\le\lfloor\sqrt{k}\rfloor$. In particular, $S>0$ by the choice of $k_0$.

Assume that the X3C instance is negative, and consider any nonzero $\mathbf v(y,w)\in\mathcal M$. We use the checker identity only when $w$ is a root of unity and the coefficient representative of $w^{-1}y$ with coefficient sum $h$ is binary. In the other cases, geometric norm bounds suffice to prove $\|\mathbf v(y,w)\|^2>S$.

\begin{samepage}
\smallskip
\noindent\emph{Case 1: $w=0$.}
Then $0\ne y\in\mathfrak a$.  \Cref{lem:checker-U-properties,thm:cyclotomic-descent} give
\[
 \|\mathbf v(y,0)\|^2
 =\|Uy\|^2
 \ge Q^2\|y\|^2
 \ge Q^2L>S,
\]
since $S/Q^2=(39/40)L+o(L)<L$.
\end{samepage}

\smallskip
\noindent\emph{Case 2: $w$ is a root of unity.}
Since $y-wc\in\mathfrak a$ and $w\mathfrak a=\mathfrak a$, we have $u:=w^{-1}y\in c+\mathfrak a$. Multiplication by $w$ is an isometry, so the identities $\mathbf v(y,w)=w\mathbf v(u,1)$ and $\|w\|^2=q-1$ give $\|\mathbf v(y,w)\|^2=\|Uu-V\|^2+Q^2F$.  By \cref{lem:coset-syndrome}, $u$ has a unique integral coefficient representative $\mathbf x\in\Z^{\F_q}$ with coefficient sum $h$.

\begin{enumerate}[label=(\alph*)]
\item If $\mathbf x\in\{0,1\}^{\F_q}$, the coefficient-sum condition shows that $\mathbf x$ has weight $h$.  Define $\xi_j:=x_{\alpha_j}$ for $j\in[n]$.  Since the instance is negative, $\|\mathbf A\boldsymbol\xi-\mathbf1_M\|_2^2\ge1$.  Thus \cref{thm:checker-score} yields
\[
 \|\mathbf v(y,w)\|^2
 \ge Q^2(B_0+2\kappa+F)=S+\kappa Q^2>S.
\]

\item If $\mathbf x\notin\{0,1\}^{\F_q}$, then \cref{lem:coset-syndrome} gives $\|u\|^2\ge H+2q$.  Dropping the negative terms in \cref{eq:checker-baseline} gives $B_0-H\le(q+1)/(32W^3)$.  Since $\kappa+\nu(H+2q+F)=O(q/\sqrt{k})=o(q)$,
\[
\begin{aligned}
 (1-\nu)(H+2q+F)-S/Q^2
 &=2q-(B_0-H)-\kappa-\nu(H+2q+F)\\
 &\ge2q-\frac{q+1}{32W^3}-o(q)>0.
\end{aligned}
\]
The last inequality uses $W\ge3$ and the choice of $k_0$. Applying \cref{eq:global-stability-exact} to $u$ and $1$ therefore gives
\[
 \|\mathbf v(y,w)\|^2
 \ge Q^2(1-\nu)(H+2q+F)>S.
\]
\end{enumerate}

\smallskip
\noindent\emph{Case 3: $w$ is not a root of unity and $0<\|w\|^2\le3q$.}
Since $y-wc\in\mathfrak a$, the coset-separation bound in \cref{thm:robust-center} gives
\[
 \|y\|^2\ge\dist(wc,\mathfrak a)^2>L/10.
\]
Because $w$ is not a root of unity, part~(ii) of \cref{lem:sparse-tools} also gives $\|w\|^2\ge2q-4$.  Hence \cref{eq:global-stability-exact} yields
\[
 \|\mathbf v(y,w)\|^2
 >Q^2(1-\nu)\bigl(L/10+\Gamma^2(2q-4)\bigr)
 =Q^2\bigl(L+o(L)\bigr)>S.
\]

\smallskip
\noindent\emph{Case 4: $w$ is not a root of unity and $\|w\|^2>3q$.}
Since $F=(18/40)L+o(L)$ and $S/Q^2<L$, we have $3F>L>S/Q^2$.  The second coordinate alone gives
\[
 \|\mathbf v(y,w)\|^2
 \ge Q^2\Gamma^2\|w\|^2
  >3Q^2F>S.
\]
These cases exhaust every nonzero module vector.  Thus $\lambda_1(\mathcal M)^2>S$.
\end{proof}

\subsection{Encoding length and running time}
\label{subsec:encoding}

We choose $q$ and bound the encoding length and running time of the reduction.

Over $\Z[\zeta_3]$, the module generated by $(2,0)$ and $(0,2)$ has squared minimum $8$.  If $M=0$, the reduction outputs this module with squared threshold $8$ as a fixed YES instance. If $M>0$ and $\mathbf A$ has a zero row, the reduction outputs this module with squared threshold $1$ as a fixed NO instance. For all remaining instances, $M,n\ge1$ and every row of $\mathbf A$ is nonzero.

\paragraph{Choosing $k$ and $q$.}
Recall that $k_0$ and $C_0$ are the fixed absolute integer constants from \cref{sec:binary-representatives}. Let $\ell$ be the bit length of the given X3C instance, and put
\[
 k:=\max\{k_0,(n+2)^2\},\qquad Y:=(2k)^{C_0}.
\]
Since $n=O(\ell)$ for this explicit encoding, $Y=O((n+2)^{2C_0})$. By Bertrand's postulate~\cite{Erdos32}, the interval $[Y,2Y]$ contains an odd prime. The reduction scans the odd integers in this interval and applies deterministic primality testing~\cite{AgrawalKayalSaxena04}. It finds a prime $Y\le q\le2Y$ with $q=\ell^{O(1)}$ in deterministic polynomial time.  There are at most $Y+1=\ell^{O(1)}$ candidates. The choice of $k$ gives $n+2\le\lfloor\sqrt{k}\rfloor$. \Cref{lem:sidon} and \cref{eq:incidence-count,eq:row-shifts} give $\beta_M=O(n^3)=O(k^{3/2})$, so the same sufficiently large $k_0$ ensures $q>4\beta_M$.

\paragraph{Generator coefficients.}
The sum defining $g$ has $2W$ nonzero coefficients, and the sum defining $V$ has total absolute coefficient sum $2W$. Expanding $U=Q+g^2$ gives an integral coefficient representative whose total absolute coefficient sum is at most $Q+4W^2$. The expansions of $(1-\zeta)^k$ and $h+(1-\zeta)^{\lfloor k/2\rfloor}$ give representatives of $\pi^k$ and $c$ with total absolute coefficient sums at most $2^k$ and $h+2^{\lfloor k/2\rfloor}$, respectively.  Hence $U\pi^k$ and $Uc-V$ have integral coefficient representatives whose coefficients are bounded in absolute value by $(Q+4W^2)(h+2^k)+2W$. Converting $\sum_{a=0}^{q-1}c_a\zeta^a$ to the power basis gives coefficients $c_a-c_{q-1}$ for $0\le a<q-1$. The maximum coefficient bit length therefore increases by at most one and is $O(k+\log q+\log n)$. The second coordinate $-Q\Gamma$ has $O(\log k+\log n)$ bit length, by \cref{eq:F-S-nu-exact} and the definition of $Q$.

\paragraph{Running time.}
The output is computed from $\mathbf A$ and the explicit center $c$. By \cref{eq:F-S-nu-exact}, $S$ has $O(\log q+\log k+\log n)$ bits. The generators and $S$ are computed by an integer square-root computation for $\Gamma$ and power-basis arithmetic on $q-1$ integer coefficients. Since $k$ and $q$ are polynomially bounded in $\ell$, the coefficient bounds above give a deterministic polynomial-time construction with polynomial encoding length. Admissible representatives serve only as witnesses in the completeness proof and need not be found by the reduction.

\subsection{Proof of the main theorem}

\begin{proof}[Proof of \cref{thm:main-intro}]
Every output of the reduction is a full-rank free $\mathcal{O}_K$-submodule of $\mathcal{O}_K^2$ and hence has module rank two.  As a lattice, each output has $\Z$-rank $2(q-1)$. For $\NP$-hardness, \cref{thm:exact-X3C-encoding} gives the YES--NO equivalence for instances with $M,n\ge1$ and no zero row in $\mathbf A$. Together with the fixed YES and NO outputs in \cref{subsec:encoding}, this proves the correctness of the reduction on all X3C instances.  The reduction is deterministic, runs in polynomial time, and outputs $(q,\mathbf m_1,\mathbf m_2,S)$.  The two module generators lie in $\mathcal{O}_K^2$ and have nonzero determinant.

A certificate for a positive instance is a nonzero $\mathbf y\in\mathcal O_K^2$ in power-basis coordinates with $\|\mathbf y\|^2\le S$. By \cref{lem:cyclotomic-form}, every coefficient has absolute value at most $\sqrt S$, so the certificate has polynomial bit length. The verifier first checks that each ring coordinate of the input generators is given by $q-1$ integer coefficients and verifies the remaining input conditions in polynomial time. It then solves $\mathbf y=z_1\mathbf m_1+z_2\mathbf m_2$ over $K$ and checks that $z_1,z_2\in\mathcal O_K$ to verify module membership. Exact power-basis arithmetic performs this check and computes $\|\mathbf y\|^2$ in polynomial time.  Thus the problem belongs to $\NP$, completing the proof of \cref{thm:main-intro}.
\end{proof}

\begin{proof}[Proof of \cref{cor:search-hardness}]
Given an X3C instance, apply the reduction from the proof of \cref{thm:main-intro} to obtain $(q,\mathbf m_1,\mathbf m_2,S)$, and query a search-$\SVP$ oracle on $(q,\mathbf m_1,\mathbf m_2)$. It returns a shortest nonzero vector $\mathbf y\in\mathcal M$ in power-basis coordinates.  These coordinates have polynomial bit length by \cref{lem:cyclotomic-form} and $\|\mathbf y\|^2\le\|\mathbf m_1\|^2$. Compute $\|\mathbf y\|^2$ using \cref{lem:cyclotomic-form}. By the correctness of the reduction, the X3C instance is positive exactly when $\|\mathbf y\|^2\le S$. Thus at most one oracle call decides X3C in polynomial time, which proves the claimed Turing hardness.
\end{proof}

\paragraph{AI disclosure.}
GPT-5.6 Sol Ultra discovered the proof through an iterative conversation conducted in a single session. The discussion began by asking whether the locally dense lattices that Bennett and Peikert construct from Reed--Solomon codes for
approximate $\SVP$~\cite{BennettPeikert23} could be adapted to module lattices.  The model first found a randomized reduction by combining this construction with Wan's point-count estimates~\cite{wan2026np}, and then developed the deterministic reduction proved in this paper. In the second version, the authors removed the restriction
$q\equiv3\pmod4$ and simplified the checker construction to avoid
lengthy Gauss-sum calculations. The authors independently verified every mathematical claim and proof, checked every use of the cited references, simplified and refined the argument, and wrote the final manuscript.  The authors take full responsibility for the paper.

\small
\begingroup
\bibliographystyle{alpha}
\bibliography{references}
\endgroup

\end{document}